\documentclass[10pt,onecolumn,draftclsnofoot,journal]{IEEEtran}
\IEEEoverridecommandlockouts
\makeatletter
\def\ps@headings{%
\def\@oddhead{\mbox{}\scriptsize\rightmark \hfil \thepage}%
\def\@evenhead{\scriptsize\thepage \hfil \leftmark\mbox{}}%
\def\@oddfoot{}%
\def\@evenfoot{}}
\makeatother
\usepackage{url}
\usepackage{nomencl}
\usepackage{graphics,latexsym, graphicx}
\usepackage{amsmath,amssymb,amsthm, amstext}
\usepackage{amsmath}
\usepackage{dsfont}
\usepackage{epsf}
\usepackage{subfigure}
\usepackage{algorithm}
\usepackage{algorithmic}
\usepackage{url}
\usepackage{color}
\newtheorem{theorem}{Theorem}

\newtheorem{lemma}{Lemma}
\newtheorem{corollary}{Corollary}
\newtheorem{remark}{Remark}
\newtheorem{proposition}{Proposition}

\DeclareMathOperator*{\argmax}{arg\,max}

\newcommand{\brho}{\mbox{\boldmath$\rho$}}
\newcommand{\CoR}{\mbox{Co}(\mathcal{R})}
\newcommand{\mD} {\mathcal{D}}
\newcommand{\mW} {\mathcal{W}}

\newcommand{\exit}{j:(n,j) \in \mathcal{L}}
\newcommand{\mS}{\mathcal{S}}

\newcommand{\hA}{\hat{A}_{nf}^{(d)}(t)}
\newcommand{\hD}{\hat{D}_{nf}^{(d)}(t)}

\newcommand{\be}{\begin{eqnarray}}
\newcommand{\ee}{\end{eqnarray}}
\newcommand{\ben}{\begin{eqnarray*}}
\newcommand{\een}{\end{eqnarray*}}

\newcommand{\expect}[1]{{\mathbb E} \Bigl[ #1\Bigr]}
\newcommand{\expectS}[1]{{\mathbb E_{\mathcal{S}(t)}} \Bigl[ #1\Bigr]}
\newcommand{\prob}[1]{{\mathbb P} \left( #1\right)}
\newcommand{\cG}{\mathcal{G}}
\newcommand{\cV}{\mathcal{V}}
\newcommand{\cE}{\mathcal{E}}

\begin{document}
\title{Flow-Level Stability of Wireless Networks: Separation of Congestion Control and Packet Scheduling}
\author{\IEEEauthorblockN{Javad Ghaderi, Tianxiong Ji, R. Srikant\\}
\IEEEauthorblockA{Department of ECE and Coordinated Science Lab.\\ University of Illinois at Urbana-Champaign
%\{jghaderi, tji2, rsrikant\}@illinois.edu
\thanks{The research was supported in part by ARO MURIs W911NF-07-1-0287 and W911NF-08-1-0233 and AFOSR MURI FA 9550-10-1-0573.}
}}

\maketitle
\begin{abstract}
It is by now well-known that wireless networks with file arrivals and departures are stable if one uses $\alpha$-fair congestion control and back-pressure based scheduling and routing. In this paper, we examine whether $\alpha$-fair congestion control is necessary for flow-level stability. We show that stability can be ensured even with very simple congestion control mechanisms, such as a fixed window size scheme which limits the maximum number of packets that are allowed into the ingress queue of a flow. A key ingredient of our result is the use of the difference between the logarithms of queue lengths as the link weights. This result is reminiscent of results in the context of CSMA algorithms, but for entirely different reasons.
\end{abstract}

\section{Introduction}\label{sec:intro}
In order to operate wireless systems efficiently, scheduling algorithms are needed to facilitate simultaneous transmissions of different users.
Scheduling algorithms for wireless networks have been widely studied since Tassiulas and Ephremides \cite{eph}
proposed the \textit{max weight} algorithm for single-hop wireless networks and its extension to multihop networks using the notion of \textit{back-pressure} or \textit{differential backlog}. The back-pressure algorithm (and hence, the max weight algorithm)
is throughput optimal in the sense that it can stabilize
the queues of the network for the largest set of arrival rates possible without knowing the actual arrival rates. The back-pressure algorithm works under very general
conditions but it does not consider \textit{flow-level dynamics}. It considers \textit{packet-level dynamics} assuming that there is a fixed set of
users/flows and packets are generated by each flow according to some stochastic process. In real networks, flows arrive randomly to the network, have only a finite amount of data, and
depart the network after the data transfer is completed.
Moreover, there is no notion of congestion control in the back-pressure algorithm while most modern
communication networks use some congestion control mechanism for fairness purposes or to avoid excessive congestion inside the network \cite{LNS08}.

There is a rich body of literature on the packet-level stability of scheduling algorithms, e.g., \cite{eph}, \cite{Neely}, \cite{KM01}, \cite{srikant}. Stability of wireless networks under flow-level dynamics
has been studied in, e.g., \cite{LNS08}, \cite{LPYCP07}, \cite{MS10}. Here, by stability, we mean
that the number of flows in the network and the queue sizes
at each node in the network remain finite. To achieve flow-level stability, these works use a specific form of congestion control based on \textit{$\alpha$-fair policies}; specifically, (a) the rates at which flows/files
generate packets into their ingress queues maximize the sum-utility where each user has a utility function of the form $U(x)=x^{1-\alpha}/(1-\alpha)$ for some $\alpha > 0$ where $x$ is the flow rate, and (b) the scheduling of packets in the network
is performed based on the max wight/back-pressure algorithm.
%Reference \cite{LNS08} uses only MAC-layer information in making the packet scheduling decisions but can only show \textit{rate stability} for $\alpha >1$.
%Reference \cite{MS10} proves that queues in the network are indeed stable for $\alpha \neq 1$ by showing that the network Markov chain is positive recurrent, however, the implicit assumption is that algorithm
%can fully observe the dynamics of queues for different flows while, in practice, the scheduler is implemented as part of the MAC layer and can thus,
%use only the MAC-layer queue lengths.
%In this paper, we use MAC-layer scheduling to achieve stability and show that the specific type of $\alpha$-fair congestion control policies is not necessary and a border class of policies can be considered that are still throughput optimal.
%

When there are file/flow arrivals and departures,
%%Upon arrival of a file, a TCP connection is established which regulates the injection of packets to MAC layer. The TCP congestion window size is controlled by a congestion control mechanism.
%The scheduling algorithm determines which links can transmit packets at each time instant.
%When the transmission of a file ends, its corresponding TCP connection is closed and the file departs the system.
if the scheduler has access to the total queue length information at nodes, then it can use max weight/back-pressure algorithm to achieve throughput optimality, but this information is not typically available to the scheduler because it is implemented as part of the MAC layer. Moreover, without congestion control, queue sizes at different nodes could be widely different. This could lead to long periods of unfairness among flows.
%because long flows will get priority over short flows, causing short flows to experience high latencies.

%In \cite{linshrsri08}, a cross-layer algorithm has been shown to be throughput-optimal; however, this result requires a particular congestion-control algorithm.
%%
Therefore, we need to use congestion control to provide better QoS. With congestion control, only a few packets from each file are released to the MAC layer at each time instant, and scheduling is done based on these MAC layer packets. However, prior work requires that a specific form of congestion control (namely, ingress queue-length based rate adaptation based on $\alpha$-fair utility functions) has to be used. \textit{Here we show that, in fact, very general window flow-control mechanisms are sufficient to ensure flow-level stability. The result suggests that ingress queue-based congestion control is more important than $\alpha$-fairness to ensure network stability}, when congestion control is used in conjunction with max weight scheduling/routing.
%We have shown similar results in \cite{javad2} for single-hop flows and, in this paper, we extend the results to the multihop scenario.

In establishing the above result, we have used the max weight algorithm with link weights which are \textit{log-differentials} of MAC-layer queue lengths, i.e., the weight of a link $(i,j$) is chosen to be in the form of $\log(1+q_i)-\log(1+q_j)$ where $q_i$ and $q_j$ are MAC-layer queue lengths at nodes $i$ and $j$.
Shorter versions of the results presented here appeared earlier in \cite{GTS12, javad, tianxiong}.
The use of logarithmic functions of queue lengths naturally suggests the use of a CSMA-type algorithm to implement the scheduling algorithm in a distributed fashion \cite{shah, shah2,ghaderi}. The main difference here is that the weights are $\log$-differential of queue lengths rather than $\log$ of queue lengths themselves. We show that the stability results for CSMA without time-scale separation can be extended to the model in this paper with log-differential of queue lengths as weights, and the type of congestion control mechanisms considered here.

At this point, we comment on the differences between our paper and a related model considered in \cite{bonfeu10}. In \cite{bonfeu10}, throughput-optimal scheduling algorithms have been derived for a connection-level model of a wireless network assuming that each link has access to the number of files waiting at the link. Here, we only use MAC-layer queue information which is readily available. Further, \cite{bonfeu10} assumes a time-scale separation between CSMA and the file arrival-departure process. Such an assumption is not made in this paper.

The rest of the paper is organized as follows.
In Section \ref{sec:system}, we describe our models for the wireless network, file arrivals, and Transport and MAC layers.
We propose our scheduling algorithm in Section \ref{sec:algorithm}. Section \ref{sec:proof} is devoted to the formal statement about the throughput-optimality of the algorithm and its proof. In Section \ref{sec: distributed}, we consider the distributed implementation of our algorithm and Section \ref{conclusion} contains conclusions. The appendices at the end of the paper contain some of the proofs.

\section{System Model}\label{sec:system}
\subsection*{Model of wireless network}
Consider a multihop wireless network consisting of a set of nodes $\mathcal{N}=\{1,2,..,N\}$ and a set of links $\mathcal{L}$ between the nodes. There is a link from $i$ to $j$, i.e.,
$(i,j) \in \mathcal{L}$, if transmission from $i$ to $j$ is allowed. Let $\mathbf{\mu}=[\mu_{ij}: (i,j) \in \mathcal{L}]$ be the rates according to which
links can transmit packets. Let $\mathcal{R}$ denote the set of available rate vectors (or transmission schedules) $\mathbf{r}=[r_{ij}: (i,j) \in \mathcal{L}]$.
Note that each transmission schedule $\mathbf{r}$ corresponds to a set of node power assignments chosen by the network. Also let $\mbox{Co}(\mathcal{R})$ denote the
convex hall of $\mathcal{R}$ which corresponds to time-sharing between different rate vectors. Hence, in general, $\mathbf{\mu} \in \CoR$.

There are a set of users/source nodes $\mathcal{U} \subseteq \mathcal{N}$ and each user/source transfers data to a destination over a fixed route in the network\footnote{The final results can be extended to case when each source has multiple destinations or to the cases of multi-path routing and adaptive routing. Here, to expose the main features, we have considered a simpler model.}. For a user/source $u \in \mathcal{U}$, we use $d(u) (\neq u)$ to denote its destination. Let $\mathcal{D}:=d(\mathcal{U})$ denote the set of destinations.

We consider a time-slotted system. At each time slot $t$,
new files can arrive at the source nodes and scheduling decisions must be made to deliver the files to destinations in multihop fashion along fixed routes.
We use $a_s(t)$ to denote the number of files that arrive at source $s$ at time $t$ and assume that the process $\{a_s(t); s \in \mathcal{U}\}_{\{t=1,2, \cdots \}}$ is iid over time and independent across users with rate $[\kappa_s; s \in \mathcal{U}]$ and has bounded second moments.
Moreover, we assume that there are $K$ possible file types where the files of type $i$ are geometrically distributed with mean $1/\eta_i$ packets. The file arrived at source $s$ can belong to type $i$ with probability $p_{si}$, $i=1,2,..,K$. Our motivation for selecting such a model is due to the \textit{large variance distribution} of file sizes in the Internet. It is believed that, see e.g.,~\cite{self}, that most of bytes are generated by long files while most of the flows are short flows. By controlling the probabilities $p_{si}$, for the same average file size, we can obtain distributions with very large variance.
%Moreover, files that arrive at user $s$ have the same size distribution independent of other users. Let $[\bar{\eta}_s;s \in \mathcal{U}]$ denote the vector of mean file sizes.
%We further assume that all file sizes are bounded from above by a constant $B$.
Let $m_s=\sum_{i=1}^K{p_{si}/\eta_i}$ denote the mean file size at node $s$, and define the work load at source $s$ by $\rho_s=\kappa_s m_s$. Let $\brho=[\rho_s: s \in \mathcal{U}]$ be the vector of loads.
\subsection*{Model of Transport and MAC layers}
Upon arrival of a file at a source Transport layer, a TCP-connection is
established that regulates the injection of packets into the
MAC layer. Once transmission of
a file ends, the file departs and the corresponding TCP-connection
will be closed. The MAC-layer is responsible for making the scheduling decisions to deliver the MAC-layer packets to their destinations over their corresponding routes. Each node has a fixed routing table that determines the next hop for each destination.

At each source node, we index the files according to their arriving order such that the index 1 is given to the earliest file. This
means that once transmission of a file ends, the indices of the remaining files are updated such that indices again start from 1
and are consecutive. Note that the indexing rule \textit{is not} part of the algorithm implementation and it is used here only for the
purpose of analysis. We use $\mathcal{W}_{sf}(t)$ to denote the TCP congestion window size for file $f$ at source $s$ at time $t$. Hence, $\mathcal{W}_{sf}$ is a time-varying sequence which changes as
a result of TCP congestion control. If the congestion window
of file $f$ is not full, TCP will continue injecting packets from
the remainder of file $f$ to the congestion window until file $f$
has no packets remaining at the Transport layer or the congestion window becomes full.
%(See Figure \ref{layer})
%\begin{figure}
%    \begin{center}
%        \includegraphics[width=3 in]{layer}
%    \end{center}
%        \caption{Transport/MAC layers: $W_1$, $W_2$, and $W_3$ show the congestion window sizes corresponding to files $1$, $2$, and $3$.}
%        \label{layer}
%\end{figure}
We consider \textit{ingress queue-based congestion control} meaning that when a packet of congestion window
departs the ingress queue, it is replaced with a new packet from its corresponding file at the Transport layer.
It is important to note that the MAC layer does not know
the number of remaining packets at the Transport layer, so scheduling decisions have to be made based on the MAC-layers information only.
It is reasonable to assume that $1 \leq \mathcal{W}_{sf}(t) \leq \mathcal{W}_{cong}$, i.e., each file has at least one packet waiting to be transferred and all congestion window sizes are bounded from above by a constant $\mathcal{W}_{cong}$.
\subsection*{Routing and queue dynamics}
At the MAC layer of each node $n \in \mathcal{N}$, we consider separate queues for the packets of different destinations. Let $q_n^{(d)}$, $d \in \mathcal{D}$, denote the packets of destination $d$ at the MAC-layer of $n$. Also let $\mathbf{R}^{(d)}_{N \times N}$ be the routing matrix corresponding to packets of destination $d$ where $R^{(d)}_{ij}=1$ if the next hop of node $i$ for destination $d$ is node $j$, for some $j$ such that $(i,j) \in \mathcal{L}$, and $0$ otherwise. Routes are acyclic meaning that each packet eventually reaches its destination and leaves the network. A packet of destination $d$ that is transmitted from $i$ to $j$ is removed from $q^{(d)}_i$ and added to $q^{(d)}_j$. Packet that reaches its destination is removed from the network. Note that packets in $q_n^{(d)}$ could be generated at node $n$ itself (if $n$ is a source with destination $d$) or belong to other sources that use $n$ as an intermediate relay along their routes to their destinations.

For the analysis, we also use $Q_n^{(d)}$ (\textit{with capital $Q$}) to denote the total per-destination queues, i.e., $Q_n^{(d)}$ represents the packets of destination $d$ at node $n$, in its MAC or Transport layer.

For each node $n$, the MAC (or total) per-destination queues $q_{n}^{(d)}$ (or $Q_{n}^{(d)}$) fall into three cases: (i) $n$ is source and $d(n)$ is its destination, (ii) $n$ is a source but $d \neq d(n)$, and (iii) $n$ is not a source. In the case (i), it is important to distinguish between the MAC-layer queue and the total queue associated with $d(n)$, i.e., $Q_n^{(d(n))} \neq q_n^{(d(n))}$, because of the existing packets of destination $d(n)$ at the Transport layer of $n$. However, $Q_n^{(d)}=q_n^{(d)}$ holds for all $d \in \mathcal{D}\backslash d(n)$ in case (ii), and for all destinations in case (iii).

Let $z_{ij}(t)$ denote the number of packets transmitted over link $(i,j) \in \mathcal{L}$ at time $t$. Then, the total-queue dynamics for a destination $d$, at each node $n$, is given by
\ben \label{eq:queue dynamics}
Q_n^{(d)}(t+1)& =& Q_n^{(d)}(t)-\sum_{j=1}^N R^{(d)}_{nj}z_{nj}^{(d)}(t) +\sum_{i=1}^N R^{(d)}_{in}z_{in}^{(d)}(t)+A^{(d)}_n(t),
\een
where $A^{(d)}_n(t)$ is the total number of packets for destination $d$ that new files bring to node $n$ at time slot $t$. To express one formula for the queue dynamics in all three cases, (i), (ii), and (iii), we can write $\expect{A^{(d)}_n(t)}=\rho^{(d)}_n$, where $\rho^{(d)}_n:=\rho_n$ in case (i) and $\rho^{(d)}_n:=0$ otherwise.

Let $x_{ij}^{(d)}$ denote the scheduling variable
that shows the rate at which the packets of destination $d$ can be forwarded over the link $(i,j)$. Note that $
z_{ij}^{(d)}(t)=\min \left\{x_{ij}^{(d)}, q^{(d)}_i(t) \right\},
$
because $i$ cannot send more than its queue content at each time.
%the total-queue dynamics are given by
%\be \label{eq:queue dynamics}
%Q_n^{(d)}(t+1)& =& \left(Q_n^{(d)}(t)-\sum_{j=1}^N x_{nj}^{(d)}(t)\right)^{+} \nonumber \\
%&&+\sum_{i=1}^N z_{in}^{(d)}(t)+A^{(d)}_n(t),
%\ee
%where $(\cdot)^+ :=\max\{0,\cdot\}$, $A^{(d)}_n(t)$ is the total number of packets of destination $d$ that new files bring to node $n$ at time $t$, and $z_{in}(t)$ is the number of packets received over the link $(i,n) \in \mathcal{L}$.

The capacity region of the network $\mathcal{C}$ is defined as the set of all load vectors $\brho$
that under which the total-queues in the network can be stabilized. Note that under our connection-level model, stability of total-queues will imply that the number of files in the network is also stable.
It is well-known \cite{Neely} that a vector $\brho$ belongs to $\mathcal{C}$ if and only if there exits a transmission rate vector $\mu \in \CoR$ such that
\begin{align*}
&\mu^{(d)}_{ij} \geq 0; \ \forall d \in \mathcal{D} \mbox{ and } \forall (i,j) \in \mathcal{L}, \\
&\rho^{(d)}_n-\sum_{j=1}^N R^{(d)}_{nj}\mu_{nj}^{(d)}+\sum_{i=1}^N R^{(d)}_{in}\mu_{in}^{(d)} \leq 0; \ \forall d \in \mathcal{D} \mbox{ and } \forall n \neq d,\\
&\sum_{d \in \mD} \mu_{ij}^{(d)} \leq \mu_{ij}; \  \forall (i,j) \in \mathcal{L}.
\end{align*}

\section{Description of Scheduling Algorithm}\label{sec:algorithm}
The algorithm is essentially the back-pressure algorithm \cite{eph} but it only uses the MAC-layer information.
The key step in establishing the optimality of such an algorithm is using an appropriate weight function of the MAC-layer queues instead of using the total queues. In particular,
consider a \textit{log-type} function
\be \label{function}
g(x):=\frac{\log(1+x)}{h(x)},
\ee
where $h(x)$ is an arbitrary increasing function which makes
$g(x)$ an increasing concave function. Assume that $h(0) > 0$
and $g(x)$ is continuously differentiable on $(0, \infty)$: For example, $h(x)=\log(e+\log(1+x))$ or $h(x)=\log^\theta(e+x)$ for some $0 <\theta < 1$.
%Let $C(n)=\{j: (n,j) \in \mathcal{L}\}$ denote the set of neighbors of node $n$.
For each link $(i,j)$ with $R^{(d)}_{ij}=1$, define
\be \label{diff log}
w_{ij}^{(d)}(t):=g(q_i^{(d)}(t))-g(q_j^{(d)}(t)).
\ee
Note that if $\{d\in \mD :R^{(d)}_{ij}=1\}=\emptyset$, then we can remove the link $(i,j)$ from the network without reducing the capacity region since no packets are forwarded over it. So without loss of generality, we assume that $\{d\in \mD :R^{(d)}_{ij}=1\}\neq \emptyset$, for every $(i,j) \in \mathcal{L}$.
Then the scheduling algorithm is as follows:

At each time $t$:
\begin{itemize}
\item Each node $n$ observes the MAC-layer queue sizes of itself and its next hop, i.e., for each $d\in \mathcal{D}$, it observes $q_n^{(d)}$ and $q_j^{(d)}$ for a $j$ such that $R^{(d)}_{ij}=1$.

\item For each link $(i,j)$, define a weight
\be \label{eq: link weight}
w_{ij}(t):= \max_{d\in \mD :R^{(d)}_{ij}=1} w_{ij}^{(d)}(t),
\ee
and
\be \label{eq: dtilde}
\tilde{d}^*_{ij}(t)=\argmax_{d\in \mD :R^{(d)}_{ij}=1} w_{ij}^{d}(t).
\ee

\item The network needs to find the optimal rate vector $\tilde{x}^* \in \mathcal{R}$ that solves
\be \label{eq: xstartilde}
\tilde{x}^*(t)=\argmax_{r \in \mathcal{R}} \sum _{(i,j) \in \mathcal{L}}r_{ij}w_{ij}(t).
\ee
\item Finally, assign $x_{ij}^{(d)}(t)=\tilde{x}^*_{ij}$ if $d={\tilde{d}}^*_{ij}(t)$, and zero otherwise (break ties at random).
\end{itemize}

\section{System Stability} \label{sec:proof}
In this section, we analyze the system and prove its stability under the algorithm described in Section \ref{sec:algorithm}. The following theorem states our main result.
\begin{theorem}\label{theorem}
For any $\brho$ strictly inside $\mathcal{C}$, the scheduling algorithm in Section \ref{sec:algorithm},
can stabilize the network independent
of transport-layer ingress queue-based congestion control mechanism (as long as the minimum window
size is one and the window sizes are bounded) and the (nonidling)
service discipline used to transmit packets from active
nodes.
\end{theorem}
\begin{remark}
Theorem \ref{theorem} holds even when $h \equiv 1$ in (\ref{function}), however, for the distributed implementation of the algorithm in Section \ref{sec: distributed}, we need $g$ to grow slightly slower than $\log$.
\end{remark}
Theorem \ref{theorem} shows that it is possible to design the ingress queue-based congestion controller regardless of the scheduling algorithm implemented in the core network. This will allow using different congestion control mechanisms at the edge of the network for different fairness or QoS considerations without need to change the scheduling algorithm implemented at internal routers of the network. As we will see, a key ingredient of such decomposition is to use difference between the logarithms of queue lengths, as in (\ref{diff log}), for the link weights in the scheduling algorithm. The rest of this section is devoted to the proof of Theorem \ref{theorem}.
\subsection*{Order of events}
Since we use a discrete-time model, we have to specify the order in which files/packets arrive and depart, which we do below:
\begin{enumerate}

\item At the beginning of each time slot, a scheduling decision is made by the scheduling algorithm. Packets depart from the MAC layers of scheduled links.

\item File arrivals occur next. Once a file arrives, a new TCP connection is set up for that file with an initial pre-determined congestion window size.

\item For each TCP connection, if the congestion window is not full, packets are injected into the MAC layer from the Transport layer until the window size is fully used or there is no more packets at the Transport layer.
\end{enumerate}
We re-index the files at the beginning of each time slot because some files might have been departed during the last time
slot.
\subsection*{State of the system}
%Define the state of node $n$ as
%\ben
%\mS_n(t)=\big\{(q_n^{(d)}(t),\mathcal{I}_n^{(d)}(t)): d \in \mD , (U_{nf}(t), \mathcal{W}_{nf}(t)): 1 \leq f \leq N_n(t) \big\}
%\een
%where $N_n(t)$ is the number of existing files at node $n$ at time $t$, $U_{nf}(t)$ is the number of remaining packets of file $f$ at node $n$ at time $t$, and $\mathcal{W}_{nf}(t)$ is its corresponding congestion window size. Obviously, if $n$ is not a source node, then we can remove $(U_{nf}, \mathcal{W}_{nf})$ from the description of $\mS_n$.
%$\mathcal{I}_n^{(d)}(t)$ denotes the information required for the specific type of service discipline in MAC-layer queues. For example, if the service discipline is \textit{FIFO} (First In-First Out), then $\mathcal{I}_n^{(d)}(t)$ is the ordering of packets in $q_n^{(d)}(t)$, or if the service discipline is random, no $\mathcal{I}_n^{(d)}(t)$ is needed.
Define the state of node $n$ as
\ben
\mS_n(t)&=&\bigl\{(q_n^{(d)}(t),\mathcal{I}_n^{(d)}(t)): d \in \mD, (\xi_{nf}(t), \mathcal{W}_{nf}(t),\sigma_{nf}(t)): 1 \leq f \leq N_n(t)\bigr\},
\een
where $N_n(t)$ is the number of existing files at node $n$ at the beginning of time slot $t$, $\sigma_{nf}(t) \in \{1/\eta_1, \cdots , 1/\eta_K\} $ is its mean size (or type), and $\mathcal{W}_{nf}(t)$ is its corresponding congestion window size. Note that $\sigma_{nf}(t)$ is a function of time only because of re-indexing
since a file might change its index from slot to slot. $\xi_{nf}(t)$ is an indicator function of whether file $f$ has still packets in the Transport layer, i.e., if $U_{nf}(t)$ is the number of remaining packets of file $f$ at node $n$, then
$$\xi_{nf}(t)=\mathds{1} \{U_{nf}(t) > \mathcal{W}_{nf}(t)\},$$
 thus $\xi_{nf}(t) = 1$, if the last packet of file $f$ has not been injected to the MAC layer of node $n$, and $\xi_{nf}(t)=0$, if there is no remaining packets of file $f$ at the Transport layer of node $n$. Obviously, if $n$ is not a source node, then we can remove $(\xi_{nf}, \mathcal{W}_{nf},\sigma_{nf})$ from the description of $\mS_n$. $\mathcal{I}_n^{(d)}(t)$ denotes the information required about $q_n^{(d)}(t)$ to serve the MAC-layer packets which depends on the specific service discipline implemented in MAC-layer queues. In the rest of the paper, we consider the case of FIFO (\textit{First In-First Out}) service discipline in MAC-layer queues. In this case, $\mathcal{I}_n^{(d)}(t)$ is simply the ordering of packets in $q_n^{(d)}(t)$ according to their entrance times. As it will turn out from the proof, the system stability will hold for any none-idling service discipline.
%e state at node $n$ is equivalent to
%\ben
%\mS_n(t)=\{Q_n^{(d)}(t): d \in \mD , \{U_{nf}(t)\}_{f=1}^{q_n(t)}, \mathcal{I}_n^{(d)}(t): d \in \mD\}
%\een
Define the state of the system to be $\mS(t)=\{\mS_n(t): n \in \mathcal{N}\}$. Now, given the scheduling algorithm in section \ref{sec:algorithm}, and our traffic model, $\mS(t)$ evolves as a discrete-time Markov chain.

\begin{remark}
We only require that the congestion window dynamics could be described as a function
of queue lengths of the network so that the network Markov
chain is well-defined. Even in the case that the congestion
window is a function of the delayed queue lengths of the
network up to T time slots before, due to the feedback
delay of at most T from destination to source, the network
state could be modified, to include the queues up to T time
slots before, so that the same proof technique still applies.
\end{remark}
Next, we analyze the \textit{Lyapunov drift} to show that the network Markov chain is positive recurrent and, as a result, the number of files in the system and queue sizes are stable.

\subsection*{Lyapunov analysis}
%\begin{theorem}\label{theorem}
%For any $\brho$ strictly inside $\mathcal{C}$, the scheduling algorithm, in Section \ref{sec:algorithm},
%can stabilize the network independent
%of Transport-layer ingress queue-based congestion control mechanism (as long as the minimum window
%size is one and the window sizes are bounded) and the (nonidling)
%service discipline used to transmit packets from active
%nodes.
%\end{theorem}
%The rest of this section is devoted to the proof of the theorem.
Define $\bar{Q}^{(d)}_n(t):= \expect{Q^{(d)}_n(t)|\mathcal{S}_n(t)}$ to be the expected total queue length at node $n$ given the state $\mathcal{S}_n(t)$. Then, if $n$ is a source, and $d$ is its destination,
\begin{eqnarray}
\label{eq:q_l_expression}
\bar{Q}^{(d)}_n(t) = q^{(d)}_n(t)+ \sum_{f=1}^{N_n(t)} \Bigl[ \sigma_{nf}(t) \xi_{nf}(t) \Bigr].
\end{eqnarray}
Otherwise, if $d \neq d(n)$ or $n$ is not a source, then $\bar{Q}^{(d)}_n(t) = q^{(d)}_n(t)$. Note that given the state $\mS(t)$, $\bar{Q}_n^{(d)}$ is known.

The dynamics of $\bar{Q}^{(d)}_n(t)$ involves the dynamics of $q^{(d)}_n(t)$, $\xi_n(t)$, and $N_n(t)$, and, thus, it consists of:

\begin{itemize}
\item [(i)]departure of MAC-layer packets
\item[(ii)]new file arrivals (if $n$ is a source)
\item[(iii)]arrival of packets from previous hops that use $n$ as an intermediate relay to forward packets to their destinations
\item[(iv)]injection of packets into the MAC layer (if $n$ is a source), and
\item [(v)] departure of files from the Transport layer (if $n$ is a source).
\end{itemize}
 Hence,
\be \label{eq:expect_q_dynamics}
\bar{Q}^{(d)}_n(t+1) & = & \bar{Q}^{(d)}_n(t) - \sum_{j=1}^NR^{(d)}_{nj} z_{nj}^{(d)}(t) + \bar{A}^{(d)}_n(t) \nonumber \\
& +& \sum_{i=1}^NR^{(d)}_{in} z_{in}^{(d)}(t) + \hat{A}^{(d)}_n(t) -\hat{D}^{(d)}_n(t),
\ee
where $\bar{A}^{(d)}_n(t) =\sum_{f=N_n(t)+1}^{N_n(t)+a_n(t)} \sigma_{nf}(t)$ is the expected number of packet arrivals due to new files,
%, given by $$\bar{A}^{(d)}_n(t) =\sum_{f=N_n(t)+1}^{N_n(t)+a_n(t)} \sigma_{nf}(t),$$
 $\hat{A}_n^{(d)}(t)$ is the total number of packets injected into the MAC layer to fill up the congestion window after scheduling and new file arrivals, and $\hat{D}_n^{(d)}(t)=\sum_{f=1}^{N_n(t)+a_n(t)} \sigma_{nf}(t) I_{nf}(t)$ is the Transport-layer  ``expected packet departure'' because of the MAC-layer injections.
%given by $$\hat{D}_n^{(d)}(t) = \sum_{f=1}^{N_n(t)+a_n(t)} \sigma_{nf}(t) I_{nf}(t).$$
Here, $I_{nf}(t)=1$ indicates that the last packet of file $f$ \textit{leaves the Transport layer during time slot $t;$} otherwise, $I_{nf}(t)=0$. To notice the difference between the indicators $I_{nf}(t)$ and $\xi_{nf}(t)$, consider a specific file and assume that its last packet enters the Transport layer at time slot $t_0$, departs the Transport layer during time slot $t_1$ and departs the MAC layer during time slot $t_2$, then its corresponding indicator $I$ is $1$ at time $t_1$ and is $0$ for $t_0 \leq t < t_1$ and $t_1< t \leq t_2$, while its indicator $\xi$ is $0$ for all time $t_1 \leq t \leq t_2$, and $1$ for $t_0 \leq  t < t_1$.

Note that $\expect{\bar{A}^{(d)}_n(t)} = \rho^{(d)}_n$ is the mean packet arrival rate at node $n$ for destination $d$.
Let $B^{(d)}_n(t) := \hat{A}_n^{(d)}(t) - \hat{D}_n^{(d)}(t)$, and
%then we rewrite (\ref{eq:expect_q_dynamics}) as
%\be
%\bar{Q}^{(d)}_n(t+1) &=& \bar{Q}^{(d)}_n(t) - \sum_{j=1}^NR^{(d)}_{nj} z_{nj}^{(d)}(t) + B^{(d)}_n(t)  \nonumber \\
%&&+ \bar{A}^{(d)}_n(t) +\sum_{i=1}^NR^{(d)}_{in} z_{in}^{(d)}(t).
%\ee
define $\expectS{\cdot}:=\expect{\cdot|\mS(t)}$. It should be clear that when $n$ is a source but $d \neq d(n)$, or when $n$ is not a source, $ \bar{A}^{(d)}_n(t)=\hat{A}^{(d)}_n(t) =\hat{D}^{(d)}_n(t)=B^{(d)}_n(t)\equiv 0$. Let $r_{max}$ denote the maximum link capacity over all the links in the network. Then Lemma~\ref{lemma:B_l} characterizes the first and second moments of $B^{(d)}_n(t)$.
\begin{lemma} \label{lemma:B_l}
For the process $\{B^{(d)}_n(t)\}$,
\begin{itemize}
\item [(i)]$\expectS {B^{(d)}_n(t) } = 0$.
\item [(ii)] Let $\eta_{min}=\min_{1 \leq i \leq K} \eta_i$, then $\expectS{B^{(d)}_n(t)^2 } \le  (\kappa_n + N^2r^2_{max} ) \max \{ \mW_{cong}^2, 1/\eta_{min}^2 \}.$
%\item [(ii)]$\expectS{B^{(d)}_n(t)^2 } \le  (\kappa_n + N^2r^2_{max} ) /\eta_{min}^2$.
\end{itemize}
\end{lemma}
Therefore, we can write
\ben \label{eq: queue dynamics}
\bar{Q}^{(d)}_n(t+1)&=& \bar{Q}^{(d)}_n(t) - \sum_{j=1}^N R^{(d)}_{nj} z_{nj}^{(d)}(t) + \tilde{A}_n^{(d)}(t) +\sum_{i=1}^NR^{(d)}_{in} z_{in}^{(d)}(t),
\een
where $\tilde{A}_n^{(d)}(t):=\bar{A}_n^{(d)}(t)+B^{(d)}_n(t)$. Note that $\tilde{A}_n^{(d)}(t)$ has mean $\rho_n^{(d)}$ and finite second moment.

Let $G(u):=\int_0^{u}g(x)dx$ for the function $g$ defined in (\ref{function}). Then $G$ is a strictly convex function. Consider a Lyapunov function
$$
V(\mS(t))=\sum_{n=1}^N\sum_{d \in \mathcal{D}}G(\bar{Q}_n^{(d)}(t)).
$$
Let $\Delta V(t) :=V(\mS(t+1))-V(\mS(t))$,
%i.e.,
%is given by
%\ben
%\Delta V(t)=\sum_{n=1}^N\sum_{d \in \mD}\left(G(\bar{Q}_n^{(d)}(t+1))-G(\bar{Q}_n^{(d)}(t))\right),
%\een
then, using convexity of $G$, we get
\ben
\Delta V(t)\leq \sum_{n=1}^N\sum_{d \in \mD}g(\bar{Q}_n^{(d)}(t+1))\left(\bar{Q}_n^{(d)}(t+1)-\bar{Q}_n^{(d)}(t)\right).
\een
Using the concavity of $g$ and the fact that $g^\prime  \leq 1$, we have
\ben
|g(\bar{Q}_n^{(d)}(t+1))-g(\bar{Q}_n^{(d)}(t))| \leq |\bar{Q}_n^{(d)}(t+1)-\bar{Q}_n^{(d)}(t)|.
\een
Furthermore, observe that, based on (\ref{eq: queue dynamics}),
%Note that
\ben \label{eq:Qdiff}
|\bar{Q}_n^{(d)}(t+1)-\bar{Q}_n^{(d)}| \leq  \tilde{A}_n^{(d)}(t)+N r_{max}.
\een
%Hence, because $g$ is strictly increasing,
%\ben
%g(\bar{Q}_n^{(d)}(t+1)) &\leq & g\Big(\bar{Q}_n^{(d)}(t)+\tilde{A}_n^{(d)}(t)+N r_{max}\Big) \\
%& \leq & g(\bar{Q}_n^{(d)}(t)) +(\tilde{A}_n^{(d)}(t)+N r_{max})
%\een
%where the last inequality follows from concavity of $g$ and the fact that $g^\prime  \leq 1$.
%Next, because $g$ is strictly increasing, we have
%\ben
%\Delta V(t)&\leq& \sum_{n=1}^N\sum_{d \in \mD} \Bigl\{g\Big(\bar{Q}_n^{(d)}(t)+\tilde{A}_n^{(d)}(t)+N r_{max}\Big)\times\\
%&&\Big(\bar{Q}_n^{(d)}(t+1)-\bar{Q}_n^{(d)}(t)\Big)\Bigr\}.
%\een
%Next, using concavity of $g$, we get
Hence,
\ben
\Delta V(t)\leq  \sum_{n=1}^N\sum_{d \in \mD}g(\bar{Q}_n^{(d)}(t))(\bar{Q}_n^{(d)}(t+1)-\bar{Q}_n^{(d)}(t))+\sum_{n=1}^N\sum_{d \in \mD}(\tilde{A}_n^{(d)}(t)+N r_{max})^2.
\een
Define, $u^{(d)}_n(t):=\max \left\{ \sum_{j=1}^N R^{(d)}_{nj}x_{nj}^{(d)}(t) - q^{(d)}_n(t), 0\right\},$ to be the wasted service for packets of destination $d$, i.e., when $n$ is included in the schedule but it does not have enough packets of destination $d$ to transmit. Then, we have
\ben
 \Delta V(t) & \leq & \sum_{n=1}^N\sum_{d \in \mD}\Big\{g(\bar{Q}_n^{(d)}(t))\Big[\sum_{i=1}^N R^{(d)}_{in} x_{in}^{(d)}(t)+\tilde{A}_n^{(d)}(t)-\sum_{j=1}^NR^{(d)}_{nj} x_{nj}^{(d)}(t)\Big]\Big\}\\
 &+& \sum_{n=1}^N\sum_{d \in \mD}g(\bar{Q}_n^{(d)}(t))u_n^{(d)}(t)+ \sum_{n=1}^N\sum_{d \in \mD}(\tilde{A}_n^{(d)}(t)+N r_{max})^2.
\een

%where
%$$\xi(\bar{Q}_n^{(d)}, \{x_{ij}^{(d)}\}):=(\bar{Q}_n^{(d)}-\sum_{\exit} x_{nj}^{(d)})_{+}-(\bar{Q}_n^{(d)}-\sum_{\exit} x_{nj}^{(d)}).$$
%Note that if $\bar{Q}_n^{(d)}(t) > \sum_{\exit} x_{nj}^{(d)}(t)$, then $\xi(\bar{Q}_n^{(d)}, \{x_{ij}^{(d)}\}) =0$. If $\bar{Q}_n^{(d)} \leq \sum_{\exit} x_{nj}^{(d)}$, then
%$0 \leq \xi(\bar{Q}_n^{(d)}, \{x_{ij}^{(d)}\}) \leq \sum_{\exit} x_{nj}^{(d)}(t) \leq N r_{max}$. Hence,
%\ben
%\Delta V(t) &\leq & \sum_{n=1}^N\sum_{d \in \mD}g(\bar{Q}_n^{(d)}(t))[\sum_{\enter} x_{in}^{(d)}(t)+\tilde{A}_n^{(d)}(t)-\sum_{\exit} x_{nj}^{(d)}(t)] \\
% & + & N^3r_{max}g(Nr_{max})+\sum_{n=1}^N\sum_{d \in \mD}(\tilde{A}_n^{(d)}(t)+N r_{max})^2.
%\een
Taking the expectation of both sides, given the state at time $t$ is known, yields
\ben
\expectS{\Delta V(t)} &\leq& \sum_{n=1}^N\sum_{d \in \mD}\Bigl\{g(\bar{Q}_n^{(d)}(t))\mathbb{E}_{\mS(t)}[\rho_n^{(d)} +\sum_{i=1}^NR^{(d)}_{in} x_{in}^{(d)}(t)-\sum_{j=1}^NR^{(d)}_{nj} x_{nj}^{(d)}(t)] \Bigr\}\\
& + & \expectS{\sum_{n=1}^N\sum_{d \in \mD}g(\bar{Q}_n^{(d)}(t))u_n^{(d)}(t)}+C_1,
\een
%where $C= N^3r_{max}f(Nr_{max})+\expect{\sum_{n=1}^N\sum_{d \in \mD}(a_n^{(d)}(t)B+N r_{max})^2} < \infty$ because $\expect{a_n^{(d)}(t)^2}<  \infty$.
where
$
C_1= \expect{\sum_{n=1}^N\sum_{d \in \mD}(\tilde{A}_n^{(d)}(t)+N r_{max})^2} < \infty,
$
because $\expect{\tilde{A}_n^{(d)}(t)^2}<  \infty$.
\begin{lemma}
\label{lemma:bound_third_term_geo}
There exists a positive constant $C_2$ such that, for all $\mS(t)$,
$$\sum_{n=1}^N\sum_{d \in \mD} \expectS{g(\bar{Q}_n^{(d)}(t))u_n^{(d)}(t)}\le C_2.$$
\end{lemma}
Using Lemma \ref{lemma:bound_third_term_geo} and changing the order of summations, we have
\ben \label{eq: drift1}
\expectS{\Delta V(t)} &\leq&  \sum_{n=1}^N\sum_{d \in \mD}g(\bar{Q}_n^{(d)}(t))\rho_n^{(d)} - \expectS{\sum_{(i,j) \in \mathcal{L}}\sum_{d \in \mD} x_{ij}^{(d)}(t)(g(\bar{Q}_i^{(d)}(t))-g(\bar{Q}_j^{(d)}(t)))}+C_1+C_2.
\een
Recall that the link weight that is actually used in the algorithm is based on the MAC-layer queues as in (\ref{diff log})-(\ref{eq: link weight}).
%\be \label{mac weight}
%w_{ij}(t)=\max_{d \in \mD} w_{ij}^{(d)}(t),
%\ee
%where
%\be \label{mac weight d}
%w_{ij}^{(d)}(t):=g(q_i^{d}(t))- g(q_j^{d}(t)).
%\ee
For the analysis, we also define a new link weight based on the state as
\be \label{state weight}
W_{ij}(t)=\max_{d \in \mD: R^{(d)}_{ij}=1} W_{ij}^{(d)}(t),
 \ee
where, for a link $(i,j) \in \mathcal{L}$ with $ R^{(d)}_{ij}=1$,
 \be \label{state weight d}
 W_{ij}^{(d)}(t):=g(\bar{Q}_i^{d}(t))- g(\bar{Q}_j^{d}(t)).
 \ee
Then, the two types of link weights only differ by a constant as stated by the following lemma.
\begin{lemma} \label{lemma}
Let $W_{ij}(t)$ and $w_{ij}(t)$, $(i,j) \in \mathcal{L}$, be the link weights defined by (\ref{state weight})-(\ref{state weight d}) and (\ref{diff log})-(\ref{eq: link weight}) respectively. Then at all times
$$|W_{ij}(t)-w_{ij}(t)| \leq \frac{\log(1+1/\eta_{min})}{h(0)}.$$
\end{lemma}
\begin{proof}
Recall that, at each node $n$, for all destinations $d \neq d(n)$, we have $\bar{Q}_n^{d}=q_n^{d}$. If $d=d(n)$ is the destination of $n$, then $\bar{Q}_n^{d}$ consists of:
(i) packets of $d$ received from upstream flows that use $n$ as an intermediate relay, and
(ii) MAC-layer packets received from the files generated at $n$ itself. Since $1 \leq \mathcal{W}_{nf}(t) \leq \mW_{cong}$, the number of files with destination $d$ that are generated at node $n$ or have packets at node $n$ as an intermediate relay, is at most $q_n^{(d)}(t)$. Therefore, it is clear that
$$q_n^{d} \leq \bar{Q}_n^{d} \leq q_n^{d} + q_n^{d}\frac{1}{\eta_{min}}.$$
Hence, for all $n$ and $d$, using a log-type function, as the function $g$ in (\ref{function}), yields
\be \label{ineq 1}
g(q_n^{d}) \leq g(\bar{Q}_n^{d}) &\leq & g\left(q_n^{d} (1+1/{\eta_{min}})\right) \nonumber \\
 & \leq & \frac{\log\left((1+q_n^{d}) (1+1/{\eta_{min}})\right)}{h(q_n^{d} (1+1/{\eta_{min}}))}  \nonumber \\
 & \leq & g(q_n^{d})+\frac{\log(1+1/\eta_{min})}{h(0)}.
\ee

It then follows that, $\forall d \in \mD$, and $\forall (i,j) \in \mathcal{L}$ with $R^{(d)}_{ij}=1$,
%\ben
%W_{ij}^{(d)}-w_{ij}^{(d)} &=& (g(\bar{Q}_i^{d})- g(\bar{Q}_j^{d})) - (g(q_i^{d})- g(q_j^{d}))\\
%& =& (g(\bar{Q}_i^{d})- g(q_i^{d})) - (g(\bar{Q}_j^{d})- g(q_j^{d}))
%\een
%which, using (\ref{ineq 1}), shows that
\be \label{ineq 2}
|W_{ij}^{(d)}-w_{ij}^{(d)}|\leq \log(1+1/\eta_{min})/h(0).
\ee
Let $d_{ij}^*:=\argmax_{d: R_{ij}^{(d)}=1} W_{ij}^{(d)}$ and $\tilde{d}^*_{ij}$ as in (\ref{eq: dtilde}).
Then, using (\ref{ineq 2}), we have that
$$w_{ij} \geq  w_{ij}^{(d_{ij}^*)}\geq W_{ij}-\log(1+1/\eta_{min})/h(0),$$
and, $$W_{ij} \geq W_{ij}^{(\tilde{d}_{ij}^*)} \geq w_{ij}-\log(1+1/\eta_{min})/h(0).$$ This concludes the proof.
\end{proof}
Let $x^*(t)$ be the max weight schedule based on weights $\{W_{ij}(t): (i,j) \in \mathcal{L}\}$, i.e.,
\be \label{eq: xstar}
x^*(t)=\argmax_{x \in \mathcal{R}} \sum _{(i,j) \in \mathcal{L}}x_{ij}W_{ij}(t).
\ee
Note the distinction between ${x}^*$ and $\tilde{x}^*$ as we used $\tilde{x}^*(t)$ in (\ref{eq: xstartilde}) to denote the Max Weight schedule based on MAC-layer queues.
Then, the weights of the schedules $\tilde{x}^*$ and $x^*$ differ only by a constant for all queue values as we show next.
First note that, from definition of $x^*$,
\be
\sum_{(i,j) \in \mathcal{L}}x^*_{ij}W_{ij}(t)-\sum _{(i,j) \in \mathcal{L}}\tilde{x}^*_{ij}W_{ij}(t) \geq 0.
\ee
Next, we have
\be
\sum_{(i,j) \in \mathcal{L}}x^*_{ij}W_{ij}(t)-\sum _{(i,j) \in \mathcal{L}}\tilde{x}^*_{ij}W_{ij}(t)&=& \sum _{(i,j) \in \mathcal{L}}x^*_{ij}W_{ij}(t)-\sum_{(i,j) \in \mathcal{L}}x^*_{ij}w_{ij}(t) \label{first1}\\
&&+  \sum_{(i,j) \in \mathcal{L}}x^*_{ij}w_{ij}(t)-\sum _{(i,j) \in \mathcal{L}}\tilde{x}^*_{ij}w_{ij}(t) \label{second2}\\
&&+ \sum _{(i,j) \in \mathcal{L}}\tilde{x}^*_{ij}w_{ij}(t) -\sum _{(i,j) \in \mathcal{L}}\tilde{x}^*_{ij}W_{ij}(t) \label{third3}\\
& \leq&  2 N^2r_{max} \log(1+1/\eta_{min})/h(0), \nonumber
\ee
because, by Lemma \ref{lemma}, (\ref{first1}) and (\ref{third3}) are less than $N^2r_{max}\log(1+1/\eta_{min})/h(0)$ each, and (\ref{second2}) is negative by definition of $\tilde{x}^*$.
Hence, \textit{under MAC scheduling $\tilde{x}^*$}, the Lyapunov drift is bounded as follows.
\ben
\expectS{\Delta V(t)} &\leq&  \sum_{n=1}^N\sum_{d \in \mD}\left\{g(\bar{Q}_n^{(d)}(t))\rho_n^{(d)}\right\}- \expectS{\sum_{(i,j) \in \mathcal{L}}x_{ij}^*(t)W_{ij}} + C,
\een
where $C=C_1+C_2+2N^2r_{max}\log(1+1/\eta_{min})/h(0)$.

Accordingly, using (\ref{state weight})-(\ref{state weight d}), and changing the order of summations in the right hand side of the above inequality yields
\ben
\expectS{\Delta V(t)}  &\leq& \sum_{n=1}^N\sum_{d \in \mD} \left\{g(\bar{Q}_n^{(d)}(t))\mathbb{E}_{\mS(t)}\bigl[\rho_n^{(d)}+\sum_{i=1}^NR^{(d)}_{in} {x^*}_{in}^{(d)}(t)-\sum_{j=1}^NR^{(d)}_{nj} {x^*}_{nj}^{(d)}(t)\bigr]\right\} + C,
\een
%\ben
% & & \left.  -\sum_{\exit} {x^*}_{nj}^{(d)}(t)\bigr]\right\} + C,
%\een
where ${x^*}_{ij}^{(d)}(t)=x^*_{ij}(t)$ for $d = d^*_{ij}$ (break ties at random) and is zero otherwise.
The rest of the proof is standard. Since load $\brho$ is strictly inside the capacity region, there must exist a $\epsilon >0$ and a $\mu \in \mbox{Co}(\mathcal{R})$ such that
\be
\rho_n^{(d)}+\epsilon \leq \sum_{j=1}^NR^{(d)}_{nj} {\mu}_{nj}^{(d)}-\sum_{i=1}^N R^{(d)}_{in}{\mu}_{in}^{(d)}\ ; \forall n \in \mathcal{N} ,\forall d \in \mD.
\ee
Hence, for any $\delta>0$,
\ben
\expectS{\Delta V(t)} &\leq& \sum_{n=1}^N\sum_{d \in \mD}g(\bar{Q}_n^{(d)}(t))\left[\sum_{i=1}^NR^{(d)}_{in} {x^*}_{in}^{(d)}(t)-\sum_{j=1}^N R^{(d)}_{nj} {x^*}_{nj}^{(d)}(t)\right]\\
&& -  \sum_{n=1}^N\sum_{d \in \mD}g(\bar{Q}_n^{(d)}(t))\left[\sum_{i=1}^NR^{(d)}_{in} {\mu}_{in}^{(d)}(t)-\sum_{j=1}^N R^{(d)}_{in}{\mu}_{nj}^{(d)}(t)\right]\\
&& - \epsilon\sum_{n=1}^N\sum_{d \in \mD}g(\bar{Q}_n^{(d)}(t)) + C .
\een
But $\sum _{(i,j) \in \mathcal{L}}x^*_{ij}W_{ij}(t) \geq \sum _{(i,j) \in \mathcal{L}}\mu_{ij} W_{ij}(t)$, $\forall \mu \in \mathrm{Co}(\mathcal{R})$, hence,
\ben
\expectS{\Delta V(t)} &\leq&  -\epsilon \sum_{n=1}^N\sum_{d \in \mD}g(\bar{Q}_n^{(d)}(t)) + C
\leq - \delta,
\een
whenever $\max_{n,d}\bar{Q}_n^{(d)} \geq g^{-1}\big(\frac{C_2+\delta}{\epsilon}\big)$ or, as a sufficient condition, whenever $\max_{n,d}q_n^{(d)} \geq g^{-1}\big(\frac{C_2+\delta}{\epsilon}\big)$. Therefore, it follows that the system is stable by an extension of the Foster-Lyapunov criteria \cite{asm} (Theorem 3.1 in \cite{eph}). In particular, queue sizes and the number of files in the system are stable.
\begin{remark}
Although we have assumed that file sizes follow a mixture of geometric distributions, our results also hold for the case of bounded file sizes with general distribution. The proof argument for the latter case is obtained by minor modifications of the proof presented in this paper (see \cite{javad}) and, hence, has been omitted for brevity.
\end{remark}

%\begin{remark}
%In the case of single-hop flows, the algorithm can be implemented in a distributed manner using the CSMA algorithm \cite{ghaderi}, \cite{shah}. In fact, it is easy to observe that our stability result holds for log-wise weight functions of the form $g(.)=\log(1+.)/h(.)$ where $h$ is an arbitrary slowly increasing function. Interestingly, this choice of weight functions is exactly the same choice of weight functions under which the throughput optimality of CSMA algorithms has been established in \cite{ghaderi} without any time-scale separation assumption. Hence, this resolves the issue of distributed implementation of the algorithm. We conjecture that the CSMA algorithm can still be used for multihop networks but validation of such conjecture and providing proof for throughput optimality of CSMA in such scenario is beyond the scope of the current paper and is a possible future work.
%\end{remark}
\section{Distributed Implementation}\label{sec: distributed}
The optimal scheduling algorithm in Section \ref{sec:algorithm} requires us to find a maximum weight-type schedule at each time, i.e., needs to solve (\ref{eq: xstartilde}) at each time $t$. This is a formidable task, hence, in this section, we design a distributed version of the algorithm based on \textit{Glauber Dynamics}.

For simplicity, we consider the following criterion for successful packet reception: Packet transmission over link $(i,j) \in \mathcal{L}$ is successful if none of the neighbors of node $j$ are transmitting. Furthermore, we assume that every node can transmit to at most one node at each time, receive from at most one node at each time, and cannot transmit and receive simultaneously (over the same frequency band). This especially models the packet reception in the case that the set of neighbors of node $i$, i.e., $C(i)=\{j: (i,j)   \in \mathcal{L}\}$, is the set of nodes that are within the transmission range of node $i$ and the interference caused by node $i$ at all other nodes, except its neighbors, is negligible. Moreover, the packet transmission over $(i,j)$ is usually followed by an ACK transmission from receiver to sender, over $(j,i)$. Hence, for a \textit{synchronized data/ACK system}, we can define a \textit{Conflict Set} (CS) for link $(i,j)$ as
\be
\mathrm{CS}_{(i,j)} &= & \Big\{(a,b) \in \mathcal{L}:a \in C(j) \mbox{, or } b \in C(i),  \mbox{or }a \in \{i,j\} \mbox{, or } b \in \{i,j\} \Big\}.
\ee
This ensures that when the links in $\mathrm{CS}_{(i,j)}$ are inactive, the data/ACK transmission over $(i,j)$/ $(j,i)$ is successful.

Furthermore, for simplicity, assume that for each link $(i,j)$, $x_{ij} \in \{0,1\}$, i.e, its service rate is one packet per time slot.
% has a unit service rate, i.e., in each time slot, one packet could be successfully transmitted over a link or .
 We can capture the interference constraints by using a \textit{conflict graph} $\cG(\cV,\cE)$, where each vertex in $\cV$ is a communication link in the wireless network. There is an edge $((i,j),(a,b)) \in \cE$ between vertices $(i,j)$ and $(a,b)$ if simultaneous transmissions over communication links $(i,j)$ and $(a,b)$ are not successful. Therefore, at each time slot, the active links should form an independent set of $\cG$, i.e., no two scheduled vertices can share an edge in $\cG$. Let $\mathcal{R}$ be the set of all such feasible schedules and $|\mathcal{L}|$ denote the number of communication links in the wireless network.

%In this case, a schedule can be formally represented by a vector $x=[x_{ij}:(i,j) \in \mathcal{L}]$ such that $x_{ij} \in \{0,1\}$ and $x_{ij}+x_{ab} \leq 1$ for all $\Bigl((i,j),(a,b)\Bigr) \in \cE$. A schedule can also be represented by a set $s$ of links such that $(i,j) \in s$ if $x_{ij} = 1$ (and $(i,j) \not \in s$ if $x_{ij}=0$). Let $\mathcal{R}$ denote the set of all feasible schedules.

%At each time slot, a feasible schedule is chosen by the scheduling algorithm based on the current network information.

We say that a node is active if it is a sender or a receiver for some active link. Inactive nodes can \textit{sense the wireless medium} and know if there is an active node in their neighborhood. This is possible because we use a synchronized data/ACK system and detecting active nodes can be performed by sensing the data transmission of active senders and sensing the ACK transmission of active receivers. Hence, using such carrier sensing, nodes $i$ and $j$ know if the channel is idle, i.e., $\sum_{(a,b)\in \mathrm{CS}_{(i,j)}}x_{ab}=0$ or if the channel is busy, i.e., $\sum_{(a,b)\in \mathrm{CS}_{(i,j)}}x_{ab} \geq 1$.
\begin{remark}
For the case of single hop networks, the link weight (\ref{eq: link weight}) is reduced to $w_{ij}(t)=g(1+q_i(t))/h(q_i(t))$ where $i$ is the source and $j$ is the destination of flow over $(i,j)$. Such a weight function is exactly the one that under which throughput optimality of CSMA has been established in \cite{ghaderi}.
Next, we will propose a version of CSMA that is suitable for the general case of multihop flows and will prove its throughput optimality. The proof uses techniques originally developed in \cite{shah, shah2} for continuous-time CSMA algorithms, and adapted in \cite{ghaderi} for the discrete-time model considered here.
\end{remark}
\subsection{Basic CSMA Algorithm for Multihop Networks }
For our algorithm, based on the MAC layer information, we define a modified weight for each link $(i,j)$ as
%\be
%\tilde{w}_{ij}(t) = \max \left(w_{ij}(t),w_{min}(t)\right),
%\ee
%where $w_{ij}(t)$ is the differential weight defined in (\ref{eq: link weight}) for the centralized algorithm, and
%\be
%w_{min}(t):=\frac{\varepsilon}{2|\mathcal{L}|} w_{max}(t),
%\ee
\be \label{eq: link weight modified}
\tilde{w}_{ij}(t) = \max_{d:R_{ij}^{(d)}=1} \tilde{w}^{(d)}_{ij}(t),
\ee
where
\be
\tilde{w}^{(d)}_{ij}(t)=\tilde{g}\left(q_i^{(d)}(t)\right)-\tilde{g}\left(q_j^{(d)}(t)\right),
\ee
and,
\be
\tilde{g}\left(q_i^{(d)}(t)\right)=\max \left\{g\left(q_i^{(d)}(t)\right),g^*(t)\right\}
\ee
where the function $g$ is the same as (\ref{function}) defined for the centralized algorithm, and
\be \label{eq: gstar}
g^*(t):=\frac{\epsilon}{4N^3} g(q_{max}(t)),
\ee
where $q_{max}(t):=\max_{i,d} q^{(d)}_i(t)$ is the maximum MAC-layer queue length in the network at time $t$ and assumed to be known, and $\epsilon$ is an arbitrary small but fixed positive number.
Note that if we remove $g^*(t)$ from the above definition, then $\tilde{w}_{ij}$ is equal to $w_{ij}$ in (\ref{diff log})-(\ref{eq: link weight}).
% and $g(.)$ is a slowly increasing function such as $g(x)=\log(e+\log(1+x))$ or $g(x)=(\log(1+x))^{\theta}$ for some $0 <\theta < 1$.

Consider the conflict graph $\cG(\cV,\cE)$ of the network as defined earlier. At each time slot $t$, a link $(i,j)$ is chosen uniformly at random, with probability $\frac{1}{|\mathcal{L}|}$, then
\begin{itemize}
\item [(i)] If $\tilde{x}_{ab}(t-1)=0$ for all links $(a,b) \in \mathrm{CS}_{(i,j)}$, then
$\tilde{x}_{ij}(t)=1$ with probability $p_{ij}(t)$, and $\tilde{x}_{ij}(t)=0$ with probability $1-p_{ij}(t)$ .\\
Otherwise, $\tilde{x}_{ij}(t)$=0.
\item[(ii)] $\tilde{x}_{ab}(t)=x_{ab}(t-1)$ for all $(a,b) \neq (i,j)$.
\item [(iii)] $x^{(d)}_{ij}(t)=\tilde{x}_{ij}(t)$ if $d=\argmax_{d:R_{ij}^{(d)}=1} \tilde{w}^{(d)}_{ij}(t)$ (break ties at random), and zero otherwise.
\end{itemize}
We choose $p_{ij}(t)$ to be
\be
p_{ij}(t)=\frac{\exp(\widetilde{w}_{ij}(t))}{1+\exp(\widetilde{w}_{ij}(t))}.
\ee
It turns out that the choice of function $g$ is crucial in establishing the throughput optimality of the algorithm for general networks. The following Theorem states the main result regarding the throughput optimality of the basic CSMA algorithm.
\begin{theorem}\label{th1}
Consider any $\epsilon > 0$. Under the function $g$ specified in (\ref{function}), the basic CSMA algorithm can stabilize the network for any $\brho \in (1-3\epsilon)\mathcal{C}$, independent of Transport-layer ingress queue-based congestion control (as long as the minimum window
size is one and the window sizes are bounded) and the (nonidling)
service discipline used to serve packets of active
queues.
\end{theorem}
\subsection{Distributed Implementation}
The basic algorithm is based on Glauber-Dynamics with one site update at each time. For distributed implementation, we need a randomized mechanism to select a link uniformly at each time slot. We use the Q-CSMA idea \cite{srikant2} to perform the link selection: Each time slot is divided into a control slot and a data slot. In the control slot, nodes exchange short control messages, similar to RTS/CTS packets in IEEE 802.11 protocol, to come up with a collision-free \textit{decision schedule} $m$. In the data slot, each link $(i,j)$ that is included in the decision schedule performs the basic CSMA algorithm.
%See \cite{srikant2}, \cite{javad2} for complete details.
The control message sent from node $j$ to $i$ in time slot $t$, contains the carrier sense information of node $j$ at time $t-1$, and the vector of MAC layer queue sizes of node $j$ at time $t$, i.e, $[q^{(d)}_j(t): d \in \mD]$, to determine the weight of link $(i,j)$.

Next, we describe the mechanisms for generation of decision schedules and data transmission schedules in more detail.

\subsection*{Generation of decision schedule}
As in \cite{srikant2}, we divide the control slot into two mini slots. In the first mini slot, each node $i$ chooses one of its neighbors $j \in C(i)$ uniformly at random, then it transmits a RTD (\textit{Request-To-Decide}) packet, containing the ID(index) of node $j$, with probability $\beta_i$. If RTD is received successfully by $j$ (i.e., $j$ and none of the neighbors of $j$ transmit RTD messages), in the second mini-slot, $j$ sends a CTD (\textit{Clear-To-Decide}) packet back to $i$, containing the ID of node $i$. The CTD message is received successfully at $i$ if there is no collision with other CTD messages. Given a successful RTD/CTD exchange over the link $(i,j)$, the link $(i,j)$ will be included in the decision schedule $m$ and no link from $\mathrm{CS}_{(i,j)}$ will be included in $m$. Hence, $m$ is a valid schedule.
So each node $i$ needs to maintain the following memories:
\begin{itemize}
\item $AS_i(t)$/$AR_i(t)$: node $i$ is included in $m(t)$ as a sender/receiver for some link.
\item $ID_i(t)$: the index of the node which is paired with $i$ as a its sender (when $AR_i(t)=1$) or its receiver (when $AS_i(t)=1$).
\item $NR_i(t)/NS_i(t)$: Carrier sense by node $i$, i.e., node $i$ has an active receiver/sender in it neighborhood during data slot $t$.
\end{itemize}
\begin{algorithm}\label{alg dec}
\caption{Decision schedule at control slot $t$ }
\begin{algorithmic}[1]
\STATE For every node $i$, set $AS_i(t)=AR_i(t)=0$.
\STATE In the first mini-slot:

- $AS_i(t)=1$ with probability $\beta_i$; otherwise $AS_i(t)=0$.

-If $AS_i(t)=1$, choose a node $j\in C(i)$ uniformly at random and send a RTD to $j$ and set $ID_i(t)=j$; otherwise listen for RTD messages.
\STATE In the second mini-slot:

-If received a RTD from $j$ in the first mini-slot, send a CTD to $j$ and set $AR_i(t)=1$ and $ID_i=j$; nodes with $AS_i(t)=1$ listen for CTD messages.

-If $AS_i(t)=1$ and CTD received successfully from $ID_i(t)$, include $(i,ID_i(t))$ in $m(t)$, otherwise $AS_i(t)=0$.

\end{algorithmic}
\end{algorithm}
CTD message sent back from a node $j$ to $i$ also contains the carrier sense information of node $j$, i.e., $NR_j(t-1)$ and $NS_j(t-1)$, and the vector of MAC layer queue sizes of node $j$ at time $t$, i.e, $q^{(d)}_j(t)$
\subsection*{Generation of data transmission schedule}
After the control slot, every node $i$ knows if it is included in the decision schedule $m(t)$, as a sender, and also knows its corresponding receiver $ID_i=j$. The data transmission schedule at time $t$, i.e., $x(t)$, is generated based on $x(t-1)$ and $m(t)$. Only those links that are in $m(t)$ can change their states and the state of other links remain unchanged. A link $(i,j)$ that is included in $m(t)$, can start a packet transmission with probability $p_{ij}(t)$ only if its conflict set has been silent during the previous time slot, as in the basic CSMA algorithm.
\begin{algorithm}\label{alg}
\caption{Data transmission schedule at slot $t$}
\begin{algorithmic}[1]
%\STATE In the control slot, randomly select a decision schedule $m(t) \subseteq \mathcal{R}$ by using access probabilities $\{a_l\}_{l=1}^N$.
\STATE
- $\forall$ $i$ with $AS_i(t)=1$ and receiver $j=ID_i$,

If no links in $\mathrm{CS}_{(i,j)}$ were active in the previous data slot, i.e., $x_{ij}(t-1) = 1$ or $NR_i(t-1)=NS_j(t-1)=0$,
\begin{itemize}
\item $x_{ij}(t) = 1$ with probability $p_{ij}(t)$,
\item  $x_{ij} = 0$ with probability $\bar{p}_{ij}(t) = 1 - p_{ij}(t)$.
\end{itemize}
Else
 $x_{ij}(t) = 0$.

- $\forall (i,j) \notin m(t)$: $x_{ij}(t) = x_{ij}(t - 1)$.

\STATE In the data slot, use $x(t)$ as the transmission schedule.
\end{algorithmic}
\end{algorithm}
\subsection*{Data transmission and carrier sensing}
In the data slot, we use $x(t)$ for the data transmission. In this slot, every node $i$ will perform of the following.

$x_{ij}(t)=1$: Node $i$ will send a data packet to node $j$.

$x_{ji}(t)=1$: Node $i$ will send an ACK to node $j$ after receiving a data packet from $j$.

All other nodes are inactive and perform carrier sensing. Since the data/ACK transmissions are synchronized in our system, every inactive node $i$ will set $NS_i(t)=0$ is it does not sense any transmission during the data transmission period and set $NS_i(t)=1$ otherwise. Similarly, node $i$ will set $NR_i(t) = 0$ if it senses
no signal during the ACK transmission period and set $NR_i(t) = 1$ otherwise.
\begin{remark}
In IEEE 802.11 DCF, the RTS/CTS exchange is used
to reduce the Hidden Terminal Problem. However, even with
RTS/CTS, the hidden terminal problem can still occur, see \cite{srikant2}. Since, in our synchronized system, RTD and CTD messages are sent in two different mini-slots, this completely eliminates the hidden terminal problem.
\end{remark}
\begin{remark}
To determine the weight at each link, $q_{max}(t)$ is also needed. Instead, each node can maintain an estimate of $q_{max}(t)$ similar to the procedure suggested in \cite{shah}. In fact, it is easy to incorporate such a procedure in our algorithm because, in the control slot, each node can include its estimate of $q_{max}(t)$ in the control messages and update its estimate based on the received control messages. Then we can use Lemma $2$ of \cite{shah} to complete the stability proof. So we do not pursue this issue here. In practical networks $\frac{\epsilon}{4 N^3} \log (1+q_{max}(t)) $ is small and we can use the weight function $g$ directly, and thus, there may not be any need to know $q_{max}(t)$.
\end{remark}
\begin{corollary}\label{dis}
Under the weight function $g$ specified in Theorem \ref{th1}, the distributed algorithm can stabilize the network for any $\brho \in (1-3\epsilon)\mathcal{C} $.
\end{corollary}
%______________________________________________
%______________________________________________
\subsection{Proof of Throughput Optimality}\label{sec:csma_proofs}
Consider the basic CSMA algorithm over a graph $\cG(\cV,\cE)$. Assume that the weights are constants, i.e., the basic algorithm uses a weight vector $\mathbf{\tilde{w}}=[\tilde{w}_{ij}:(i,j) \in \mathcal{L}]$ at all times. Then, the basic algorithm is essentially an irreducible, aperiodic, and reversible Markov chain (called Glauber Dynamics) to generate the independent sets of $\cG$. So, the state space $\mathcal{R}$ consists of all independent sets of $\cG$. The stationary distribution of the chain is given by
\begin{equation}\label{stationary}
\pi(s)=\frac{1}{Z} \exp\Big(\sum_{(i,j) \in s} \tilde{w}_{ij}\Big);\ \ s \in \mathcal{R},
\end{equation}
where $Z$ is the normalizing constant.

We start with the following lemma that relates the modified link weight and the original link weight.
\begin{lemma}\label{relate1}
For all links $(i,j) \in \mathcal{L}$, the link weights (\ref{eq: link weight modified}) and (\ref{eq: link weight}) differ at most by $g^*(t)$, i.e.,
\be
|\tilde{w}_{ij}(t) - w_{ij}(t)| \leq g^*(t).
\ee
\end{lemma}
Proof is included in the appendix.
%See the appendix for the proof.
%\begin{proof}
%
%It is sufficient to prove that for all $d \in \mD$, $w^{(d)}_{ij}(t)-g^*(t) \leq \tilde{w}^{(d)}_{ij}(t) \leq w^{(d)}_{ij}(t)+g^*(t)$ as we do now.
%\ben
%\tilde{w}^{(d)}_{ij} &\leq& \max \left\{g\left(q_i^{(d)}(t)\right),g^*(t)\right\}-g\left(q_j^{(d)}(t)\right)\\
%&\leq & g\left(q_i^{(d)}(t)\right)+ g^*(t) -g\left(q_j^{(d)}(t)\right)\\
%& = & w_{ij}^{(d)}(t)+g^*(t).
%\een
%Similarly,
%\ben
%\tilde{w}^{(d)}_{ij} &\geq& g\left(q_i^{(d)}(t)\right)-\max \left\{g\left(q_j^{(d)}(t)\right),g^*(t)\right\}\\
%&\geq & g\left(q_i^{(d)}(t)\right)- g^*(t) -g\left(q_j^{(d)}(t)\right)\\
%& = & w_{ij}^{(d)}(t)-g^*(t).
%\een
%\end{proof}
The basic algorithm uses a time-varying version of the Glauber dynamics, where the weights change with time. This yields a time-inhomogeneous Markov chain but we will see that, for the choice of function $g$ in (\ref{function}), it behaves similarly to the Glauber dynamics.
\subsection*{Mixing time of Glauber dynamics}
For simplicity, we index the elements of $\mathcal{R}$ by $1,2,...,r$, where $r=|\mathcal{R}|$.
%Since the Markov chain is irreducible and
%aperiodic, a unique stationary distribution $\pi=[\pi(1), ..., \pi(r)]$ always exists.
Then, the eigenvalues of the corresponding transition matrix are ordered in such a way that
$$
\lambda_1 =1>\lambda_2 \geq ... \geq \lambda_r > -1.
$$
The convergence to steady state distribution is geometric with a rate equal to the \textit{second largest eigenvalue modulus} (SLEM) of the transition matrix \cite{pier}. In fact, for any initial probability distribution $\mu_0$ on $\mathcal{R}$, and for all $n \geq 1$,
\begin{equation}\label{mixing1}
\|\mu_0 \mathbf{P}^n-\pi\| _{\frac{1}{\pi}} \leq (\lambda^*)^n \|\mu_0-\pi\| _{\frac{1}{\pi}},
\end{equation}
where $\lambda^*=\max\{\lambda_2, |\lambda_r|\}$ is the SLEM.
Note that, by definition,
$
\|z\|_{1/\pi}=\left(\sum_{i=1}^r z(i)^2\frac{1}{\pi(i)}\right)^{1/2}.
$

%\begin{lemma}\label{mixing1}
%Let $P$ be an irreducible, aperiodic, and reversible transition matrix on the finite state space $\mathcal{R}$ with the stationary distribution $\pi$. Then, the eigenvalues of $P$ are ordered in such a way that
%$$
%\lambda_1 =1>\lambda_2 \geq ... \geq \lambda_r > -1,
%$$
%and for any initial probability distribution $\mu_0$ on $\mathcal{R}$, and for all $n \geq 1$
%\begin{equation}
%\|\mu_0 \mathbf{P}^n-\pi\| _{\frac{1}{\pi}} \leq (\lambda^*)^n \|\mu_0-\pi\| _{\frac{1}{\pi}},
%\end{equation}
%where $\lambda^*=\max\{\lambda_2, |\lambda_r|\}$ is the SLEM of $P$.
%\end{lemma}
The following Lemma gives an upper bound on the SLEM $\lambda^*$ of Glauber dynamics.
\begin{lemma}\label{mixing2}
For the Glauber Dynamics with the weight vector $\mathbf{\tilde{w}}$ on a graph $\cG(\cV,\cE)$,
$$
\lambda^* \leq 1-\frac{1}{16^{|\cV|}\exp(4|\cV|\tilde{w}_{max})},
$$
where $\tilde{w}_{max}=\max_{(i,j) \in \mathcal{L}} \tilde{w}_{ij}$.
\end{lemma}
See \cite{ghaderi} for the proof. We define the \textit{mixing time} as $T=\frac{1}{1-\lambda^*}$, so
\begin{equation} \label{mixing3}
T \leq 16^{|\mathcal{L}|}\exp(4|\mathcal{L}|\tilde{w}_{max})
\end{equation}
Simple calculation, based on (\ref{mixing1}), reveals that the amount of time needed to get close to the stationary distribution is \textit{approximately} proportional to $T$.
\subsection*{A key proposition}
At any time slot $t$, given the weight vector $\mathbf{\tilde{w}}(t)=[\tilde{w}_{ij}(t): (i,j) \in \mathcal{L}]$, the MaxWeight-type algorithm, Section \ref{sec:algorithm}, should solve $
\max_{s \in \mathcal{R}} \sum_{(i,j) \in s}\tilde{w}_{ij}(t),
$
instead, the distributed algorithm tries to simulate a distribution
\begin{equation}\label{stationary}
\pi_t(s)=\frac{1}{Z_t} \exp\Big(\sum_{(i,j) \in s} \tilde{w}_{ij}(t)\Big);\ \ s \in \mathcal{R},
\end{equation}
i.e., the stationary distribution of Glauber dynamics with the weight vector $\mathbf{\tilde{w}}(t)$ at time $t$.

Let $P_t$ denote the transition probability matrix of Glauber dynamics with the weight vector $\mathbf{\tilde{w}}(t)$.
Also let $\mu_t$ be the true probability distribution of the inhomogeneous-time chain, over the set of schedules $\mathcal{R}$, at time  $t$.
Therefore, we have $\mu_t=\mu_{t-1}P_t$. Let $\pi_t$ denote the stationary distribution of the time-homogenous Markov chain with $P=P_t$ as in (\ref{stationary}).
By choosing proper $g^*$ and $g(\cdot)$, we aim to ensure that $\mu_t$ and $\pi_t$ are close enough, i.e.,
%$$
%2 \|\pi_t-\mu_t\|_{TV} \leq \|\pi_t-\mu_t\|_{1/\pi_t} \leq \delta
%$$
$
 \|\pi_t-\mu_t\|_{TV} \leq \delta
$
for some $\delta$ arbitrary small, where
$\|\pi-\mu\|_{TV}=\frac{1}{2}\sum_{i=1}^r|\pi(i)-\mu(i)|.
$
Note that $
\|\mu-\pi\|_{\frac{1}{\pi}} \geq 2 \|\mu - \pi\|_{TV}.
$
%Let $w_{max}(t)=f(q_{max}(t))$.
%Let $g_{max}(t) = g(q_{max}(t))$.
Next, we characterize the amount of change in the stationary distribution as a result of queue evolutions.
\begin{lemma} \label{alpha ratio}
For any schedule $s \in \mathcal{R}$,
$
e^{-\alpha_t}\leq \frac{\pi_{t+1}(s)}{\pi_t(s)} \leq e^{\alpha_t},
$
where,
\begin{equation} \label{eq: alpha}
\alpha_t= 2(1+\mW_{cong})|\mathcal{L}| g^\prime\Big(g^{-1}(g^*(t+1))-1-\mW_{cong}\Big),
\end{equation}
and $\mW_{cong}$ is the maximum congestion window size.
\end{lemma}
Now, equipped with Lemmas \ref{mixing2} and \ref{alpha ratio}, we make use of the results in \cite{shah, shah2} and \cite{ghaderi} in the final proof. Specifically, we will use the following key Proposition from \cite{ghaderi} which we have included a proof for it in the appendix for completeness.
%\begin{lemma}\label{driftb}
%Given any $\delta >0$, $\|\pi_t-\mu_t\|_{TV} \leq \delta/4$ holds for all $t \geq t^*$, if
%\begin{equation}\label{alphaTb}
%\alpha_t T_{t+1} \leq \delta/16\mbox{ for all }  t > 0,
%\end{equation}
%where
%\begin{itemize}
%\item[(i)] $T_t \leq 16^{|\mathcal{L}|}\exp(4|\mathcal{L}|\tilde{w}_{max}(t))$.
%\item [(ii)] $t^*$ is the smallest $t$ such that
%\be\label{t*b}
%\frac{1}{\sqrt{\min_s \pi_{0}(s)}} \exp(-\sum_{k=t_1}^{t_1+t^*}\frac{1}{T^2_k}) \leq \delta/4.
%\ee
%\end{itemize}
%\end{lemma}
%\begin{remark}
%It is worth pointing out that Lemma \ref{driftb} is consistent with the adiabatic Theorem in quantum mechanics stating that: A physical system remains close to its instantaneous eigenstate if a given perturbation is acting on it with a rate much slower than the rate at which the system responses to the perturbation. As we will see later, in our setting, this implies that the system remains close to its instantaneous equilibrium distribution if the rate of change in sum-weights at time $t$, which is less than $\alpha_t$, is much smaller than $\frac{1}{T_{t+1}}$ where $T_{t+1}$ is the mixing time of the Markov chain with weights $W_{t+1}$.
%\end{remark}
%See the appendix of \cite{javad} for the proof of the above Lemma.
\begin{proposition}\label{drift}
Given any $\delta >0$, $\|\pi_t-\mu_t\|_{TV} \leq \delta/4$ holds when $ q_{max}(t)\geq q_{th}+t^*$, if there exists a $q_{th}$ such that
\begin{equation}\label{alphaT}
\alpha_t T_{t+1} \leq \delta/16\mbox{ whenever } q_{max}(t) > q_{th},
\end{equation}
where
\begin{itemize}
%\item [(i)] $\alpha_t=\alpha_t= 4(1+\mW_m)|\mathcal{L}| g^\prime(g^{-1}(g^*(t+1))-1-\mW_m)$
\item[(i)] $T_t \leq 16^{|\mathcal{L}|}\exp(4|\mathcal{L}|\tilde{w}_{max}(t))$
\item [(ii)] $t^*$ is the smallest $t$ such that
\be\label{t*}
\frac{1}{\sqrt{\min_s \pi_{t_1}(s)}} \exp(-\sum_{k=t_1}^{t_1+t^*}\frac{1}{T^2_k}) \leq \delta/4,
\ee
where $q_{max}(t_1)=q_{th}$.
%\begin{equation}\label{t*}
%\sum_{k=t_1:\|q(t_1)\|=q_{th}}^{t_1+t^*}\frac{1}{T^2_k} \geq \log (4/\delta)+|\mathcal{L}|(g(q_{th})+\log2)/2.
%\end{equation}
\end{itemize}
\end{proposition}
%In the above Lemma, condition (ii) is based on the upper bound of (\ref{mixing3}) and the fact that $\tilde{w}_{max}(t)\leq (1+\frac{\epsilon}{4 |\mathcal{L}|})w_{max}(t)$.
%See the appendix for the proof of the above Lemma.
In other words, Proposition \ref{drift} states that when queue lengths are large, the observed distribution of the schedules is close to the desired stationary distribution. The key idea in the proof is that the weights change at the rate $\alpha_t$ while the system responds to these changes at the rate $1/T_{t+1}$. Condition (\ref{alphaT}) is to ensure that the weight dynamics are slow enough compared to response time of the chain such that it remains close to its equilibrium (stationary distribution).
%\begin{remark}
%We will later see that, to satisfy condition (\ref{alphaT}) and to find a finite $t^*$ satisfying (\ref{t*}) in Lemma \ref{drift},
%the function $f(\cdot)$ cannot be faster than $\log(\cdot)$.
%In fact, the function $f$ must be slightly slower than $\log(\cdot)$ to
%\end{remark}
%\begin{remark}
%The above Lemma is a generalization of Lemma $12$ (Network Adiabatic Theorem) of \cite{shah}.
%Here we consider general functions $f(\cdot)$, whereas \cite{shah} considers a particular function $\log \log(\cdot)$.
%The generalization allows us to use functions which are close to $\log(\cdot)$ and perform much better than $\log \log(\cdot)$ in simulations.
%\end{remark}

We will also use the following lemma that relates the maximum queue length and the maximum weight in the network. Hence, when one grows, the other one increases as well.
\begin{lemma} \label{relate}
Let $w_{max}(t)=\max_{(i,j) \in \mathcal{L}} w_{ij}(t)$. Then
$$
\frac{1}{N}g\left(q_{max}(t)\right) \leq w_{max}(t) \leq g\left(q_{max}(t)\right).
$$
\end{lemma}
%The proof is simple and has been omitted for brevity (see \cite{javad 2}) for details).
\subsection*{Some useful results for the basic CSMA algorithm}
Roughly speaking, since the mixing time $T$ is exponential in $g(q_{max})$, $g^\prime(g^{-1}(g^*))$ must be in the form of $e^{-{g^*}}$;
otherwise it will be impossible to satisfy $\alpha_tT_{t+1}< \delta/16$ in Proposition \ref{drift} for any arbitrarily small $\delta$ as $q_{max}(t) \to \infty$.
The only function with such a property is the $\log(\cdot)$ function. In fact, $g$ must grow slightly slower than $\log(\cdot)$ to satisfy (\ref{alphaT}), and to ensure the existence of a finite $t^*$ in Lemma \ref{drift}. For example, by choosing functions that grow much slower than $\log(1+x)$, like $h(x)=\log(e+\log(1+x))$, we can make $g(x)$ behave approximately like $\log(1+x)$ for large ranges of $x$ (correspondingly, for the range of practical queue lengthes). More accurately, we state the result as the following lemma whose proof can be found in the appendix.
\begin{lemma}\label{lemma: adiabatic}
The Basic CSMA algorithm, with function $g$ as in (\ref{function}), satisfies the requirements of Proposition \ref{drift}.
\end{lemma}

Next, the following Lemma states that, with high probability, the basic CSMA algorithm chooses schedules that their weights are close to the Max Wight schedule.
\begin{lemma}\label{property}
The basic CSMA algorithm has the following property: Given any $0<\varepsilon<1$ and $0<\delta<1$, there exists a $B(\delta, \varepsilon)>0$ such that whenever $q_{max}(t) >B(\delta, \varepsilon)$, with probability larger than $1-\delta$, it chooses a schedule $s(t) \in \mathcal{R}$ that satisfies
$$
\sum_{(i,j) \in s(t)}w_{ij}(t) \geq (1-\epsilon) \max_{s \in \mathcal{R}} \sum_{(i,j) \in s}w_{ij}(t).
$$
\end{lemma}
\begin{proof}
Let $w^*(t)=\max_{s \in \mathcal{R}} \sum_{(i,j) \in s} w_{ij}(t)$ and define
$$
\chi_t:=\Big\{s \in \mathcal{R}: \sum_{(i,j) \in s}w_{ij}(t) < (1-\epsilon) w^*(t) \Big\}.
$$
Therefore, we need to show that $\mu_t(\chi_t) \leq \delta$, for $q_{max}(t)$ large enough. For our choice of $g(\cdot)$ and $g^*$, it follows from Proposition \ref{drift} that, whenever $q_{max}(t)> q_{th}+t^*$, $2 \|\mu_t-\pi_t\|_{TV} \leq \delta/2$,
and consequently, $
\sum_{s \in \mathcal{R}} \Bigl |\mu_t(s)-\pi_t(s)\Bigr| \leq \delta/2.
$
Thus,
%\begin{eqnarray*}
%\Bigl |\sum_{s \in \mathcal{\chi}_t} (\mu_t(s)-\pi_t(s)) \Bigr|
%&\leq & \sum_{s \in \mathcal{\chi}_t} \Bigl |\mu_t(s)-\pi_t(s) \Bigr| \leq  \delta/2
%\end{eqnarray*}
%Then, it can be shown that
$$
\sum_{s \in \mathcal{\chi}_t} \mu_t(s) \leq \sum_{s \in \mathcal{\chi}_t} \pi_t(s)+\delta/2.
$$
Therefore, to ensure that $\sum_{s \in \chi_t} \mu_t(s) \leq \delta$, it suffices to have $
\sum_{s \in \mathcal{\chi}_t}\pi_t(s) \leq \delta/2.
$
But, by Lemma \ref{relate1}, $\widetilde{w}_{ij}(t) \leq w_{ij}(t)+g^*(t)$,
%\begin{eqnarray*}
%\sum_{s \in \mathcal{\chi}_t} \pi_t(s) &= &\sum_{s \in \mathcal{\chi}_t} \frac{1}{Z_t}\exp(\sum_{i \in s}\widetilde{w}_i(t))
%\end{eqnarray*}
%where $ \widetilde{w}_i(t)=\leq w_i(t)+g^*(t).$
so,
\begin{eqnarray*}
\sum_{s \in \mathcal{\chi}_t} \pi_t(s) &\leq & \sum_{s \in \mathcal{\chi}_t} \frac{1}{Z_t}e^{\sum_{(i,j) \in s}{w}_{ij}(t)}e^{|s|g^*(t)}\leq \sum_{s \in \mathcal{\chi}_t} \frac{1}{Z_t}e^{(1-\varepsilon)w^*(t)}e^{|\mathcal{L}|g^*(t)},
\end{eqnarray*}
and
\ben
Z_t &=& \sum_{s \in \mathcal{R}}e^{\sum_{(i,j) \in s}\widetilde{w}_{ij}(t)} > \sum_{s \in \mathcal{R}}e^{\sum_{(i,j) \in s}(w_{ij}(t)-g^*(t))}
> e^{w^*(t)-|\mathcal{L}|g^*(t)}.
\een
Therefore,
\ben
\sum_{s \in \mathcal{\chi}_t} \pi_t(s) &\leq & 2^{|\mathcal{L}|}e^{2|\mathcal{L}|g^*(t)-\varepsilon w^*(t)},
\een
when $q_{max}(t)> q_{th}+t^*$. Note that $w^*(t)\geq w_{max}(t) \geq g(q_{max}(t))/{N}$, and $g^*(t)=\frac{\epsilon}{4N^3}g(q_{max}(t))$, so
\begin{eqnarray}
\sum_{s \in \mathcal{\chi}_t} \pi_t(s) &\leq & 2^{N^2}e^{-\frac{\epsilon}{2{N}}g(q_{max}(t))} \leq \delta/2 \nonumber
\end{eqnarray}
whenever $q_{max}(t) > B(\delta, \epsilon)$, where
\ben
B(\delta,\epsilon)=\max\left\{q_{th}+t^*, g^{-1}\Big(\frac{2{N}}{\epsilon}(N^2 \log{2}+\log{\frac{2}{\delta}})\Big)\right\}.
\een
\end{proof}
\subsection*{Throughput optimality}
Now we are ready to prove the throughput optimality for the basic CSMA algorithm. Let $x^*$ and $\tilde{x}^*$ be the optimal schedules based on total queues and MAC queues respectively, given by (\ref{eq: xstar}) and (\ref{eq: xstartilde}), and $\tilde{x}$ be the schedule generated by the basic CSMA algorithm. The proof is parallel to the argument for the throughput optimality of the centralized algorithm. Especially, the inequality (\ref{eq: drift1}) still holds, i.e.,
\be \label{drift1}
\expectS{\Delta V(t)} & \leq &  C_1+C_2+\sum_{n=1}^N\sum_{d \in \mD}g(\bar{Q}_n^{(d)}(t))\rho_n^{(d)} -  \expectS{\sum_{(i,j) \in \mathcal{L}}\sum_{d \in \mD} {x}_{ij}^{(d)}(t)W_{ij}^{(d)}} \nonumber\\
& = &  C_1+C_2+\sum_{n=1}^N\sum_{d \in \mD}g(\bar{Q}_n^{(d)}(t))\rho_n^{(d)} -  \expectS{\sum_{(i,j) \in \mathcal{L}}{\tilde{x}}_{ij}(t)W_{ij}(t)}.
\ee
Next, observe that
%Similar to the argument for the centralized algorithm,
%from definition of $x^*$, we have
%\be
%\sum_{(i,j) \in \mathcal{L}}x^*_{ij}W_{ij}(t)-\expectS{\sum _{(i,j) \in \mathcal{L}}\tilde{x}_{ij}W_{ij}(t)} \geq 0,
%\ee
%and,
\be
\sum_{(i,j) \in \mathcal{L}}x^*_{ij}W_{ij}(t)-\expectS{\sum _{(i,j) \in \mathcal{L}}\tilde{x}_{ij}W_{ij}(t)}&=& \expectS{\sum _{(i,j) \in \mathcal{L}}x^*_{ij}W_{ij}(t)-\sum_{(i,j) \in \mathcal{L}}{x}^*_{ij}w_{ij}(t)} \label{first}\\
&+&  \expectS{\sum_{(i,j) \in \mathcal{L}}x^*_{ij}w_{ij}(t)-\sum _{(i,j) \in \mathcal{L}}\tilde{x}_{ij}w_{ij}(t)} \label{second}\\
&+& \expectS{\sum _{(i,j) \in \mathcal{L}}\tilde{x}_{ij}w_{ij}(t) -\sum _{(i,j) \in \mathcal{L}}\tilde{x}_{ij}W_{ij}(t)} \label{third}
\ee
Now note that each of the terms (\ref{first}) and (\ref{third}) is less than $|\mathcal{L}|\log(1+1/\eta_{min})/h(0)$ by Lemma (\ref{lemma}). The term (\ref{second}) is bounded from above, by using Lemma \ref{property}, as follows.
\ben \label{forth}
%&&\expectS{\sum_{(i,j) \in \mathcal{L}}x^*_{ij}w_{ij}(t)-\sum _{(i,j) \in \mathcal{L}}\tilde{x}_{ij}w_{ij}(t)} \nonumber\\
\mbox{(\ref{second})} & \leq &   \sum_{(i,j) \in \mathcal{L}}x^*_{ij}w_{ij}(t) -(1-\delta)(1-\epsilon)\sum_{(i,j) \in \mathcal{L}}\tilde{x}^*_{ij}w_{ij}(t) \nonumber \\
& \leq &\sum_{(i,j) \in \mathcal{L}}x^*_{ij}w_{ij}(t) -(1-\delta)(1-\epsilon)\sum_{(i,j) \in \mathcal{L}}{x}^*_{ij}w_{ij}(t) \nonumber \\
& \leq &(1-(1-\delta)(1-\epsilon)) \sum_{(i,j) \in \mathcal{L}} {x}^*_{ij}W_{ij}(t) + |\mathcal{L}| \log(1+1/\eta_{min})/h(0),
\een
whenever $q_{max}(t) \geq B(\delta, \epsilon)$, for any $\delta >0$.
Thus, using the above bounds for terms (\ref{first}), (\ref{second}) and (\ref{third}), we get
\be \label{fifth}
\expectS{\sum _{(i,j) \in \mathcal{L}}\tilde{x}_{ij}W_{ij}(t)} &\geq & (1-\delta)(1-\epsilon)\sum_{(i,j) \in \mathcal{L}} {x}^*_{ij}W_{ij}(t) -  3|\mathcal{L}|  \log(1+1/\eta_{min})/h(0).
\ee
Using (\ref{fifth}) in (\ref{drift1}) yields
\be \label{drift2}
\expectS{\Delta V(t)} & \leq &  \sum_{n=1}^N\sum_{d \in \mD}g(\bar{Q}_n^{(d)}(t))\rho_n^{(d)} - (1-\delta)(1-\epsilon)\sum_{(i,j) \in \mathcal{L}} {x}^*_{ij}W_{ij}(t)+C_3,
\ee
where $C_3:=C_1+C_2 +3 |\mathcal{L}|  \log(1+1/\eta_{min})/h(0)$.
Using (\ref{state weight}) and rewriting the right-hand-side of (\ref{drift2}) by changing the order of summations yields
\be \label{drift3}
\expectS{\Delta V(t)}  \leq   \sum_{n=1}^N\sum_{d \in \mD}g(\bar{Q}_n^{(d)}(t))\Bigl[\rho_n^{(d)}  +(1-\delta)(1-\epsilon) \Bigl(\sum_{i=1}^NR^{(d)}_{in}{{x}^*}^{(d)}_{in}(t)-\sum_{j=1}^NR^{(d)}_{nj}{{x}^*}^{(d)}_{nj}(t)\Bigr)\Bigr]+ C_3.\nonumber
\ee
whenever $q_{max}(t) \geq B(\delta, \epsilon)$.
The rest of the proof is standard. For any load  $\brho$ strictly inside $(1-3 \epsilon)\mathcal{C}$, there must exist a $\mu \in \mbox{Co}(\mathcal{R})$ such that for all $1 \leq n \leq N$, and all $d \in \mD$,
\be \label{capacity 1}
\rho^{(d)}_n < (1-3 \epsilon) \Big(\sum_{j=1}^NR^{(d)}_{nj} {\mu}_{nj}^{(d)}-\sum_{i=1}^N R^{(d)}_{in}{\mu}_{in}^{(d)}\Big).
\ee
%and, else
%%for all other nodes $n \in \mathcal{N}\backslash \mathcal{U}$, or for sources $n \in \mathcal{U}$ but $d \neq d(n)$,
%\be \label{capacity 2}
%0 < \left(\sum_{j=1}^N R^{(d)}_{nj}{\mu}_{nj}^{(d)}-\sum_{i=1}^N R^{(d)}_{in}{\mu}_{in}^{(d)}\right).
%\ee
Let
$
{\rho^*} ={(1-3 \epsilon)} \min_{n\in \mathcal{N},d \in \mD} \left(\sum_{j} R^{(d)}_{nj}{\mu}_{nj}^{(d)}-\sum_{i}R^{(d)}_{in} {\mu}_{in}^{(d)}\right)
$
for some positive $\rho^*$. Hence,
\ben
\expectS{\Delta V(t)} &\leq&  (1-\delta)(1-\epsilon)\sum_{n=1}^N\sum_{d \in \mD} \Bigl\{g(\bar{Q}_n^{(d)}(t)) \Bigl[\sum_{i=1}^NR^{(d)}_{in}{{x}^*}^{(d)}_{in}(t)-\sum_{j=1}^NR^{(d)}_{nj}{{x}^*}^{(d)}_{nj}(t)\Bigr] \Bigr\}\\
& + & (1-3 \epsilon) \sum_{n=1}^N\sum_{d \in \mD}\Bigl\{g(\bar{Q}_n^{(d)}(t)) \Bigl[\sum_{j=1}^N R^{(d)}_{nj}{\mu}_{nj}^{(d)}-\sum_{i=1}^N R^{(d)}_{in}{\mu}_{in}^{(d)}\Bigr] \Bigr\}+ C_3 .
\een
For any fixed small $\epsilon >0$, we can choose $\delta < \epsilon /(1-\epsilon)$ to ensure $(1-\delta)(1-\epsilon) > 1-2 \epsilon$. Moreover, from definition of $x^*(t)$ and convexity of $\mathrm{Co}(\mathcal{R})$, it follows that
\be
\sum_{n=1}^N\sum_{d \in \mD}g(\bar{Q}_n^{(d)}(t))\Bigl[\sum_{j=1}^N R^{(d)}_{nj}{x^*}_{nj}^{(d)}(t)-\sum_{i=1}^NR^{(d)}_{in} {x^*}_{in}^{(d)}(t)\Bigr] \geq \sum_{n=1}^N\sum_{d \in \mD}g(\bar{Q}_n^{(d)}(t)) \Bigl[\sum_{j=1}^NR^{(d)}_{nj} {\mu}_{nj}^{(d)}-\sum_{i=1}^NR^{(d)}_{in} {\mu}_{in}^{(d)}\Bigr],
\ee
for any $ \mu \in \mathrm{Co}(\mathcal{R})$. Hence,
\ben
\expectS{\Delta V(t)} &\leq& -\epsilon \sum_{n=1}^N\sum_{d \in \mD}g(\bar{Q}_n^{(d)}(t))\Bigl[\sum_{j=1}^N {R^{(d)}_{nj}\mu}_{nj}^{(d)}-\sum_{i=1}^NR^{(d)}_{in} {\mu}_{in}^{(d)}\Bigr] + C_3 \\
& \leq& - \rho^* \frac{\epsilon}{1-3 \epsilon}\sum_{n=1}^N\sum_{d \in \mD}g(\bar{Q}_n^{(d)}(t))+C_3 \leq - \epsilon',
\een
whenever $\max_{n,d}\bar{Q}_n^{(d)} \geq g^{-1}\big(\frac{C_3+\epsilon'}{\rho^*}\frac{1-3 \epsilon}{\epsilon}\big)$ and $q_{max}(t) \geq B(\delta, \epsilon)$ or, as a sufficient condition, whenever
$$q_{max}(t) \geq \max \left\{B(\delta, \epsilon), g^{-1}\big(\frac{C_3+\epsilon'}{\rho^*}\frac{1-3 \epsilon}{\epsilon}\big)\right\}.$$
In particular, to get negative drift, $-\epsilon'$, for some positive constant $\epsilon'$, it suffices that
 $$\max _n N_n> \max\left\{g^{-1}\big(\frac{C_3+\epsilon'}{\rho^*}\frac{1-3 \epsilon}{\epsilon}\big), B(\delta, \epsilon)\right\}$$
  because $ q_{max}(t) \geq  \max _n N_n$, and $g$ is an increasing function. This concludes the proof of the main theorem.

\subsection*{Extension of the proof to the distributed implementation}
The distributed algorithm is based on multiple site-update (or parallel operating) Glauber dynamics as defined next. Consider the graph $\cG(\cV,\cE)$ as before and a constant weight vector $\mathbf{\tilde{w}}=[\tilde{w}_{ij}: (i,j) \in \mathcal{L}]$. At each time $t$, a decision schedule $m(t) \subseteq \mathcal{R}$ is selected at random with positive probability $\alpha(m(t))$. Then, for all $(i,j) \in m(t)$, we perform the regular Glauber dynamics.
%\begin{itemize}
%\item [(i)] If $x_j[t-1]=0$ for all nodes $j \in \mathcal{|\mathcal{L}|}(i)$, then
%$x_i(t)=1$ with probability $\frac{\exp(\tilde{w}_i)}{1+\exp(\tilde{w}_i)}$, or $x_i(t)=0$ with probability $\frac{1}{1+\exp(\tilde{w}_i)}$ .\\
%Otherwise, $x_i(t)=0.$
%\item[(ii)] $x_j(t)=x_j[t-1]$ for all $j \notin m(t)$.
%\end{itemize}
 Then, it can be shown that the Markov chain is reversible, it has the same stationary distribution as the regular Glauber dynamics in (\ref{stationary}), and its mixing time is almost the same as (\ref{mixing3}). In fact, the mixing time of the chain is characterized by the following Lemma.
\begin{lemma}\label{multiple}
For the multiple site-update Glauber Dynamics with the weight vector $\tilde{W}$ on a graph $\cG(\cV,\cE)$,
\begin{equation}
T \leq \frac{64^{|\cV|}}{2}\exp(4|\cV|\tilde{w}_{max}).
\end{equation}
where $\tilde{w}_{max}=\max_{i \in \cV|} \tilde{w}_i$.
\end{lemma}
See \cite{ghaderi} for the proof. The rest of the analysis is the same as the argument for the basic algorithm.
The distributed algorithm uses a time-varying version of the multiple-site update Glauber dynamics, where the weights change with time. Although the upperbound of Lemma \ref{multiple} is loose, it is sufficient to prove the optimality of the algorithm.

Finally, let $T_{data}$ and $T_{control}$ denote the lengths of the data slot and the control slot. Thus, the distributed algorithm can achieve a fraction ${T_{data}}/(T_{data}+T_{control})$ of the capacity region. In particular, in our algorithm, it suffices to allocate two short mini-slots at the beginning of the slot for the purpose of control. By choosing the data slot to be much larger than the control slot, the algorithm can approach the full capacity.

\section{Conclusions}\label{conclusion}
Design of efficient scheduling and congestion control algorithms can be decoupled by using MAC-layer queues for the scheduling of packets and using window-based congestion control mechanisms for controlling the rate at which packets injected into the network. This separation result is very appealing to the network designer. It is also important from the practical perspective because, typically, only the MAC-layer information is available to the scheduler since it is implemented as part of the MAC layer. Moreover, window-based congestion control is also more consistent with practical implementations like TCP.

\appendices
\section{Proof of Lemma \ref{lemma:B_l}}
Let $\hat{A}_{nf}^{(d)}(t)$ denote the number of packets of file $f$ injected into the MAC layer of node $n$, and $\hat{D}_{nf}^{(d)}(t) = \sigma_{nf}(t)I_{nf}(t)$ denote the expected ``packet departure'' of file $f$ from the transport layer. Let $B_{nf}(t) = \hA- \hD$ for file $f$.
\begin{proof}[Part (i)]
It suffices to show that for each individual file $1 \leq f \leq N_n(t)$,
$
\expectS{B_{nf}(t)}=0.
$
We only need to focus on files $f$ with $ \xi_{nf}(t)=1$, i.e., existing files in the Transport layer, or new files, i.e, $f \in \bigl(N_n(t)+1, N_n(t) + a_n(t) \bigr)$, because the $\expectS{B_{nf}(t)}=0$ if file $f$ has no packets in the Transport layer.

Let $\mW_{nf}^r(t)$ be the remaining window size of file $f$ at node $n$ after MAC-layer departure but before the MAC-layer injection. We want to show that, for any $w \ge 0$,
\begin{eqnarray}\label{eq:inner_transfer2}
\expectS{B_{nf}(t) \Bigl| \mW_{nf}^r(t) =w }=0,
\end{eqnarray}
then (\ref{eq:inner_transfer2}) implies $\expectS{B_{nf}(t)}=0$.

Because the number of remaining packets at the Transport layer at each time is geometrically distributed with mean size $\sigma_{nf}(t)$, the transport layer will continue to inject packets into the MAC layer with probability $\gamma_{nf}(t) = 1-1/\sigma_{nf}(t)=1-\eta_{nf}(t)$ as long as all previous packets are successfully injected and the window size is not full.

Clearly, if $w=0,$ no packet can be injected into the MAC layer. Therefore, $\hA=0$ and $\hD=0$, and (\ref{eq:inner_transfer2}) is satisfied. Next, we consider the case when $w>0$.
Let $p_w(k,j)$ denote the probability that $\hA=k$ and $I_{nf}(t)=j\in\{0,1\}$ given that $\mW_{nf}^r(t)=w$. Because transport-layer packets are injected into the MAC layer as long as the window is not full, we have $p_w(k,0)=0 \mbox{ for } k <w.$ Obviously, $p_w(k,1)=0 \mbox{ for } k>w.$

The probability that $\hA = k$ where $k < w$ directly follows the geometric distribution of the remaining packets of file $f$, i.e.,
\begin{eqnarray*}
 p_w(k,1) & = & \prob{\hA=k, I_{nf}(t) = 1 | \mW_{nf}^r(t)=w} \\
  &= & \prob{\hA=k | \mW_{nf}^r(t)=w} \\
  &= & \gamma_{nf}^{k-1}(t) (1- \gamma_{nf}(t)),
\end{eqnarray*}
for $1 \le k \leq w$.
%The probability that the window size gets full (i.e., $\hA = w$) and the last packet is injected into the MAC-layer (i.e., $I_{nf}(t) = 1$) is
%\ben
%p_w(w,1) & =& \prob{\hA=w, I_{nf}(t) = 1 | \mW_{nf}^r(t)=w} \\
%&=& \gamma_{nf}^{w-1}(1-\gamma_{nf}).
%\een
%
%The probability that the window size is full after injection is
%\begin{eqnarray*}
%P(a_{li}^{mac}[t] =K) = \sum_{k=K}^{\infty} \gamma_{li}^{k-1}(1-\gamma_{li}) = \gamma_{li}^{K-1}.
%\end{eqnarray*}
%
Note that from the definition of $I_{nf}(t),$ we have
\begin{eqnarray*}
\prob{I_{nf}(t) = 0 | \mW_{nf}^r(t)=w} = 1 - \sum_{k=1}^w p_w(k,1)  = \gamma_{nf}^w(t).
\end{eqnarray*}

Now we calculate the left-hand side of (\ref{eq:inner_transfer2}).
\begin{eqnarray*}
\expectS {B_{nf}(t) | \mW_{nf}^r(t) = w } &=& \sum_{k=1}^{w} p_w(k,1) \Bigl( k- \sigma_{nf} \Bigr)  + \prob{I_{nf}(t) = 0|\mW_{nf}^r(t)=w} w\\
&=& \sum_{k=1}^w k\gamma_{nf}^{k-1}(1- \gamma_{nf}) - (1- \gamma_{nf}^w) \sigma_{nf} + w \gamma_{nf}^w \\
&=&(1-\gamma_{nf}) \frac{d}{d \gamma_{nf}} \Bigl[ \frac{\gamma_{nf} - \gamma_{nf}^{w+1}}{1- \gamma_{nf}} \Bigr] - \frac{1 - \gamma_{nf}^w}{1- \gamma_{nf}}+ w \gamma_{nf}^w \\
&=& 0.
\end{eqnarray*}
\end{proof}
\begin{proof}[Part (ii)]
 From the definition of $B_{n}(t),$ we have $$B_n(t) = \sum_{f=1}^{N_n(t)} B_{nf}(t) + \sum_{f=N_n(t)+1}^{N_n(t)+a_n(t)}  B_{nf}(t).$$
Using the fact that new arriving files are mutually independent, and are also independent of current network state, we have
\be \label{eq:bound_B_l}
\expectS{ B_n(t)^2 } &=& \expectS{ \Bigl(\sum_{f=1}^{N_n(t)} B_{nf}(t) \Bigr)^2 }  + \expectS{\sum_{f=N_n(t)+1}^{N_n(t)+ a_n(t)}   B_{nf}(t)^2},
\ee
where we have also used the fact that $\expectS{ B_{nf}(t)}=0.$
Note that $B_{nf}(t)^2 \le \max \{ \hA^2, \hD^2\}.$ So, based on the assumption that the congestion window size is bounded by $\mW_{cong}$ and the mean file size is bounded by $1/\eta_{min},$ we have
$\expectS{ B_{nf}(t)^2 } \le \max \{ \mW_{cong}^2 ,  1/\eta_{min}^2 \}.$
Therefore, the second term in (\ref{eq:bound_B_l}) is bounded by
\begin{align} \label{eq:bound_B_l_new}
& \expectS{\sum_{f=N_n(t)+1}^{N_n(t)+ a_n(t)}   B_{nf}(t)^2} < \kappa_n  \max  \{ \mW_{cong}^2 ,  1/\eta_{min}^2 \}.
% & \le \max  \{\mW_{cong}^2 ,  1/\eta_{min}^2 \} \times \expectS {a_n(t) } \nonumber \\
%& \kappa_n  \max  \{ \mW_{cong}^2 ,  1/\eta_{min}^2 \}.
\end{align}
Next, we bound the first term in (\ref{eq:bound_B_l}). Let $\mathcal{F}_n(t)$ denote the set of files at node $n$ that are served at time $t$. Because $B_{nf}(t) = 0$ if the existing file is not served, we have
\begin{eqnarray*}
\Bigl| \sum_{f=1}^{N_n(t)} B_{nf}(t) \Bigr| &\le& \max \Bigl \{ \sum_{f \in \mathcal{F}_n(t)} \hA , \sum_{f \in \mathcal{F}_n(t)} \sigma_{nf}(t) \Bigr \} \nonumber \\
&\le& \bigl|\mathcal{F}_n(t)\bigr| \cdot \max \Bigl\{ \mW_{cong}, 1/\eta_{min} \Bigr\}.
\end{eqnarray*}

Note that $|\mathcal{F}_n(t)| \le \sum_{\exit}x_{nj}(t) \leq N r_{max}$ because the number of existing files that are served cannot exceed the sum of outgoing link capacities. Thus,
\begin{eqnarray}\label{eq:bound_B_l_existing}
\expectS { \Bigl(\sum_{f=1}^{n_l(t)} B_{lf}(t) \Bigr)^2  } \le   N^2 r^2_{max} \max \Bigl\{ \mW_{cong}^2, 1/\eta_{min}^2 \Bigr\}
\end{eqnarray}
Substituting (\ref{eq:bound_B_l_new}) and (\ref{eq:bound_B_l_existing}) into (\ref{eq:bound_B_l}) completes the proof.

%However, the upper bound of the lemma is not tight. Now we give a proof for an even tighter upper bound.
%Because $B_{ni}(t)^2 \le \max \{ \hA^2, \hD^2\},$ we have $$\expectS {B_{ni}(t)^2 } \le \max \Bigl\{ \expectS { \hA^2 }, \expectS { \hD^2 } \Bigr\}.$$ Because the number of packets of file $i$ at the transport layer is geometrically distributed with parameter $\sigma_{ni}(t)$ and the number of packets injected into the MAC-layer $\hA$ can not exceed the number of packets at the transport layer, therefore, $\expectS { \hA^2 } \le \sigma_{ni}(t)(\sigma_{ni}(t)-1) \le 1/\eta_{min}^2.$ Also note that $\expectS { \hD)^2 } \le \sigma_{ni}(t)^2 \le 1/\eta_{min}^2,$ thus we have
%$$\expectS{B_{ni}(t)^2 } \le 1/\eta_{min}^2.$$
%%
%The second term of (\ref{eq:bound_B_l}) is bounded by
%%
%\begin{eqnarray*}
%\expectS {\sum_{i=N_n(t)+1}^{N_n(t) +a_n(t)} B_{ni}(t)^2  } &\le&  1/\eta_{min}^2\cdot  \expectS { a_n(t) }\nonumber = \lambda_n  /\eta_{min}^2.
%\end{eqnarray*}
%%
%Using Cauchy-Schwarz Inequality, the first term in (\ref{eq:bound_B_l}) is bounded as follows.
%\ben
% &&\expectS{\Bigl(\sum_{i=1}^{N_n(t)} B_{ni}(t) \Bigr)^2 } = \expectS{ \Bigl(\sum_{i \in \mathcal{F}_n(t)} B_{ni}(t) \Bigr)^2  }\\
% &&\le  |\mathcal{F}_n(t)| \cdot \expectS{ \sum_{i \in \mathcal{F}_n(t)} B_{ni}(t)^2 } \le |\mathcal{F}_ n(t)|^2 / \eta_{min}^2 \\
% &&\le N^2r^2_{max} / \eta_{min}^2.
%\een
%%
%Therefore,
%$\expectS {B_n(t)^2 } \le (\lambda_n + N^2r_{max}^2)/\eta_{min}^2.$
\end{proof}
\section{Proof of Lemma \ref{lemma:bound_third_term_geo}}
 %Recall that the wasted service is $u^{(d)}_n(t) = \max\{\sum_{\exit} x_{nj}^{(d)}(t)-q_n^{(d)}(t),0 \}$. Therefore,
 Note that $u^{(d)}_n(t)=0$ if $q_n^{(d)}(t) \geq N r_{max}$, and $u^{(d)}_n(t) \leq Nr_{max}$ if $q_n^{(d)}(t) \leq N r_{max}$. In the latter case, since the congestion window size for every file is at least one, we know that there are at most $Nr_{max}$ files in transport layer of node $n$ intended for destination $d$. Hence, based on the definition of $\bar{Q}^{(d)}_n(t)$, $\bar{Q}^{(d)}_n(t) \ge Q^0 := Nr_{max}+ Nr_{max}/{\eta_{min}}$. So,
\begin{eqnarray*}
\expectS{g(\bar{Q}_n^{(d)}(t))u_n^{(d)}(t)}&=& \expectS{g(\bar{Q}_n^{(d)}(t))u_n^{(d)}(t)\mathds{1}\left\{q_n^{(d)}(t) \leq N r_{max}\right\}}  \\
&\leq& \expectS{g(\bar{Q}_n^{(d)}(t))Nr_{max} \mathds{1}\left\{q_n^{(d)}(t) \leq N r_{max}\right\} }  \\
&\leq&  Nr_{max}g(Q^0)
\end{eqnarray*}
Therefore, the result follows by choosing $C_2=N^3r_{max} g(Nr_{max}(1+ 1/\eta_{min})).$
\section{Proof of Lemma \ref{relate1}}
It is sufficient to prove that for all $d \in \mD$, $w^{(d)}_{ij}(t)-g^*(t) \leq \tilde{w}^{(d)}_{ij}(t) \leq w^{(d)}_{ij}(t)+g^*(t)$ as we do now.
\ben
\tilde{w}^{(d)}_{ij} &\leq& \max \left\{g\left(q_i^{(d)}(t)\right),g^*(t)\right\}-g\left(q_j^{(d)}(t)\right)\\
&\leq & g\left(q_i^{(d)}(t)\right)+ g^*(t) -g\left(q_j^{(d)}(t)\right)\\
& = & w_{ij}^{(d)}(t)+g^*(t).
\een
Similarly,
\ben
\tilde{w}^{(d)}_{ij} &\geq& g\left(q_i^{(d)}(t)\right)-\max \left\{g\left(q_j^{(d)}(t)\right),g^*(t)\right\}\\
&\geq & g\left(q_i^{(d)}(t)\right)- g^*(t) -g\left(q_j^{(d)}(t)\right)\\
& = & w_{ij}^{(d)}(t)-g^*(t).
\een
\section{Proof of Lemma \ref{alpha ratio}}
Note that
$$
\frac{\pi_{t+1}(s)}{\pi_t(s)}=\frac{Z_t}{Z_{t+1}}\exp\Big({\sum_{(i,j) \in s}(\widetilde{w}_{ij}(t+1)-\widetilde{w}_{ij}(t))}\Big)
$$
%or
%$$
%\exp\left(-\sum_{i \in \rho}|\widetilde{w}_i(t+1)-\widetilde{w}_i(t)|\right) \leq \frac{\pi_{t+1}(\rho)}{\pi_t(\rho)} \leq \exp\left(\sum_{i \in \rho}|\widetilde{w}_i(t+1)-\widetilde{w}_i(t)|\right).
%$$
where
\begin{eqnarray*}
\frac{Z_t}{Z_{t+1}} & =&\frac{\sum_{s \in \mathcal{R}}\exp(\sum_{(i,j)  \in s}\widetilde{w}_{ij} (t))}{\sum_{s \in \mathcal{R}}\exp(\sum_{(i,j) \in s}\widetilde{w}_{ij}(t+1))}\leq \max_{s}\exp\Big({\sum_{(i,j)  \in s}(\widetilde{w}_{ij} (t)-\widetilde{w}_{ij} (t+1))}\Big)\\
& \leq & \exp\Big({\sum_{(i,j)\in \mathcal{L}}(\widetilde{w}_{ij} (t)-\widetilde{w}_{ij} (t+1))}\Big).
\end{eqnarray*}
Let $q^*(t)$ denote $g^{-1}(g^*(t))$, and define $\widetilde{q}^{(d)}_i(t):=\max\{q^*(t), q^{(d)}_i(t)\}$.
Hence, %if $\widetilde{q}_i(t+1) \geq \widetilde{q}_i(t)$,
\begin{eqnarray*}
\widetilde{w}_{ij}^{(d)}(t+1)-\widetilde{w}_{ij}^{(d)}(t) &=& g(\widetilde{q}^{(d)}_{i}(t+1))-g(\widetilde{q}^{(d)}_j(t+1))-g(\widetilde{q}^{(d)}_{i}(t))+g(\widetilde{q}^{(d)}_j(t)) \\
&=& \left[g(\widetilde{q}^{(d)}_{i}(t+1))-g(\widetilde{q}^{(d)}_{i}(t))\right]+ \left[g(\widetilde{q}^{(d)}_{j}(t))-g(\widetilde{q}^{(d)}_{j}(t+1))\right]\\
&\leq& g^\prime(\widetilde{q}^{(d)}_i(t)) (\widetilde{q}^{(d)}_i(t+1)-\widetilde{q}_i(t))+  g^\prime(\widetilde{q}^{(d)}_j(t+1)) (\widetilde{q}^{(d)}_j(t)-\widetilde{q}_j(t+1)),
\end{eqnarray*}
where the last inequality follows from the fact that $g$ is a concave and increasing function. If we assume that link service rate is at most one and the congestion window sizes are at most $\mW_{cong}$, then for all $i \in \mathcal{N}$ and for all $d\in\mD$,
$
|\widetilde{q}^{(d)}_i(t+1)-\widetilde{q}^{(d)}_i(t)| \leq 1 + \mW_{cong}.
$
Hence,
\begin{eqnarray*}
\frac{|\widetilde{w}_{ij}^{(d)}(t+1)-\widetilde{w}^{(d)}_{ij}(t)|}{1 + \mW_{cong}} & \leq &  g^\prime(\widetilde{q}^{(d)}_i(t))+ g^\prime(\widetilde{q}^{(d)}_j(t+1))\leq  2 g^\prime(q^{*}(t+1)-1-\mW_{cong}),
\end{eqnarray*}
%
%Similarly, %if $\widetilde{q}_i(t+1) \leq \widetilde{q}_i(t)$,
%\begin{eqnarray*}
%\widetilde{w}_{ij}^{(d)}(t+1)-\widetilde{w}^{(d)}_{ij}(t) \geq (1 + \mW_m)( g^\prime(\widetilde{q}^{(d)}_i(t+1))+ g^\prime(\widetilde{q}^{(d)}_j(t)))
%\end{eqnarray*}
and thus,
\ben
\frac{\pi_{t+1}(s)}{\pi_t(s)} \leq e^{2(1+\mW_{cong})|\mathcal{L}| g^\prime({q}^*(t+1)-1-\mW_{cong})}.
\een
%\ben
%\frac{\pi_{t+1}(s)}{\pi_t(s)} &\leq& e^{(1+\mW_m)\sum_{(i,j)\in \mathcal{L}} \max _d g^\prime(\widetilde{q}^{(d)}_i(t))+ g^\prime(\widetilde{q}^{(d)}_i(t+1))}\\
%&\times & e^{(1+\mW_m)\sum_{(i,j)\in \mathcal{L}} g^\prime(\widetilde{q}^{(d)}_j(t+1))+ g^\prime(\widetilde{q}^{(d)}_j(t))}\\
%& \leq & e^{4(1+\mW_m)|\mathcal{L}| g^\prime({q}^*(t+1)-1-\mW_m)}.
%\een
%because $\widetilde{q}^{(d)}_i$ is reduced by at most one in each time slot.
Similarly,
\ben
\frac{\pi_{t}(s)}{\pi_{t+1}(s)} \leq e^{2(1+\mW_{cong})|\mathcal{L}| g^\prime({q}^*(t+1)-1-\mW_{cong})}.
\een
%$$
%|\widetilde{w}_i(t+1)-\widetilde{w}_i(t)| \leq f^\prime(\widetilde{q}_i(t))+f^\prime(\widetilde{q}_i(t+1)) \leq f^\prime(\widetilde{q}_i(t+1)-1)+f^\prime(\widetilde{q}_i(t+1))
%$$
This concludes the proof.
\section{Proof of Lemma \ref{relate}}
The second inequality immediately follows from definition of $w_{ij}$. To prove the first inequality, consider a destination $d$, with routing matrix $\mathbf{R}^{(d)} \in \{0,1\}^{N \times N}$, and let $\mathbf{w}^{(d)}=[w_{ij}^{(d)}(t): R^{(d)}_{ij}=1]$, then, based on (\ref{diff log}), we have
\ben
\mathbf{w}^{(d)}=(\mathbf{I}-\mathbf{R}^{(d)})g(\mathbf{q}^{(d)}),
\een
where $g(\mathbf{q}^{(d)})=[g(q_i^{(d)}): i \in \mathcal{N}].$ Note that every row of $\mathbf{R}^{(d)}$ has exactly one ''1`` entry except the row corresponding to $d$ which is all zero, so $(\mathbf{R}^{(d)})^N=0$. Therefore, $(\mathbf{I}-\mathbf{R}^{(d)})^{-1}=\mathbf{I}+\mathbf{R}^{(d)}+(\mathbf{R}^{(d)})^2+ \cdots $ exists and $\mathbf{I}-\mathbf{R}^{(d)}$ is nonsingular (Similar to the argument in page 222 of \cite{data}). So $
g(\mathbf{q}^{(d)})=(\mathbf{I}-\mathbf{R}^{(d)})^{-1}\mathbf{w}^{(d)}$. Let $\|\cdot \|_{\infty}$ denote the $\infty$-norm. Then we have
\ben
\|(\mathbf{I}-\mathbf{R}^{(d)})^{-1}\|_{\infty} &=&  \|\sum_{k=0}^N (\mathbf{R}^{(d)})^k\|_{\infty}
 \leq  \sum_{k=0}^N \|(\mathbf{R}^{(d)})^k\|_{\infty}
\leq  \sum_{k=0}^N \|\mathbf{R}^{(d)}\|_{\infty}^k
 \leq  N
\een
%In fact, we can rearrange the rows (relabel the nodes) to make $\mathbf{I}-\mathbf{R}^{(d)}$ a upper-triangular matrix with all ''1`` diagonal entries and all ''-1`` super-diagonal entries. Therefore all the eigenvalues of $\mathbf{I}-\mathbf{R}^{(d)}$ are one and hence $\| \mathbf{w}^{(d)} \|_2 =\| g\left(\mathbf{q}^{(d)}\right) \|_2.$
where we have used the basic properties of the matrix norm, and the fact that $\|\mathbf{R}^{(d)}\|_{\infty}=1$. Therefore,
$$ \|g(\mathbf{q}^{(d)})\|_{\infty} \leq \|(\mathbf{I}-\mathbf{R}^{(d)})^{-1}\|_{\infty}  \|\mathbf{w}^{(d)}\|_{\infty} \leq N \|\mathbf{w}^{(d)}\|_{\infty},$$
for every $d \in \mathcal{D}$. Taking the maximum over all $d \in \mathcal{D}$, and noting that $g$ is a strictly increasing function, yields the result.
\section{Proof of Lemma \ref{lemma: adiabatic}}
$h$ is strictly increasing so $h(x) \geq 1$ for all $x \geq h^{-1}(1)$. So
\begin{equation}
g^\prime(x) \leq \frac{1}{1+x}\mbox{ for } \ x \geq h^{-1}(1).
\end{equation}
The inverse of $g$ cannot be expressed explicitly, however, it satisfies
\begin{equation}\label{g-1}
g^{-1}(x)=\exp(xh(g^{-1}(x)))-1.
\end{equation}
Therefore,
\begin{eqnarray}
\alpha_t  & \leq &  \frac{2(1+\mW_{cong})|\mathcal{L}|}{g^{-1}(g^*)-\mW_{cong}}  = \frac{2(1+\mW_{cong})|\mathcal{L}|}{\exp(g^*h(g^{-1}(g^*)))-1-\mW_{cong}}.\label{help1}
\end{eqnarray}
for $g^* \geq g(1+\mW_{cong}+h^{-1}(1))$.

Next, note that
\be
T_{t+1} &\leq& 16^{|\mathcal{L}|}e^{4 |\mathcal{L}|(w_{max}+g^*)} \nonumber \\
& \leq & 16^{|\mathcal{L}|}e^{4 |\mathcal{L}|(g(q_{max})+\frac{\epsilon}{4|\mathcal{L}| N }g(q_{max}))} \nonumber \\
& \leq & 16^{|\mathcal{L}|}e^{8 |\mathcal{L}| g(q_{max})} \label{help22}.
\ee
%Hence, the conditions of Lemma \ref{drift} are satisfied if there exists a $q_{th}$ large enough such that product of (\ref{help1}) and (\ref{help22}) is less than $\delta/16$ for $\|q(t)\| \geq q_{th}$.
Consider the product of (\ref{help1}) and (\ref{help22}) and let $K:=2 (\mW_{cong}+1)|\mathcal{L}|16^{|\mathcal{L}|}$.
%\begin{equation}\label{help2}
%K\exp(4|\mathcal{L}|w_{max})\frac{1}{\exp(g^*h(g^{-1}(g^*)))-1-\mW_m} \leq \delta/16
%\end{equation}
Using (\ref{g-1}) and (\ref{eq: gstar}), the condition (\ref{alphaT}) is satisfied if
\begin{equation}\label{cond1}
Ke^{g^*[\frac{32|\mathcal{L}|N^3}{\epsilon}-h(g^{-1}(g^*))]}\left(1+\frac{1+\mW_{cong}}{g^{-1}(g^*)-\mW_{cong}}\right) \leq \delta/16.
\end{equation}
Consider fixed, but arbitrary, $|\mathcal{L}|$, $N$ and $\epsilon$. As $q_{max} \to \infty$, $g(q_{max}) \to \infty$, and consequently $g^* \to \infty$ and $g^{-1}(g^*) \to \infty$. Therefore, the exponent $\frac{32|\mathcal{L}|N^3}{\epsilon}-h(g^{-1}(g^*))$ is negative for $q_{max}$ large enough, and thus, there is a threshold $q_{th}$ such that for all $q_{max} > q_{th}$, the condition (\ref{cond1}) is satisfied.
%To be more accurate, it suffices to choose
%\begin{equation}\label{qth}
%q_{th}=f^{-1}\left(\frac{2N}{\epsilon} \times \max\left\{ \log(\frac{64N16^N}{\delta}), f(g^{-1}(\frac{16N^2}{\epsilon}))\right\}\right).
%\end{equation}
%Then, it follows from Corollary \ref{drift} that $\|\pi_t-\mu_t\|_{TV} \leq \frac{\delta}{4}$, whenever $\|q^{mac}\|> q_{th}+t^*$.
%\begin{remark}
%The assumption $g(0) \geq 1$ is not required, since, as we saw in the above analysis, only the asymptotic behavior of $g$ is important. If we choose $q_{th}$ large enough
%such that
%\begin{equation}\label{help3}
%g(f^{-1}(w_{min}[t])-1) \geq 1
%\end{equation}
%when $\|bq^{mac}\| \geq q_{th}$, then (\ref{help11}) holds and the rest of the analysis follows exactly.
%In particular, in order to get an explicit formula for $f^{-1}$, we can choose $g(x)=\log(1+x)^\theta$ for some $0 < \theta <1$.
%The weight function for such a $g$ is $f(x)=\left(\log(1+x)\right)^{1-\theta}$, and $f^{-1}$ has the closed form
%$$
%f^{-1}(x)=\exp(x^{\frac{1}{1-\theta}})-1.
%$$
%Then (\ref{qth}) yields
%\begin{equation}\label{thresh2}
%q_{th} = \exp\left( \max\left\{  \frac{2N}{\epsilon}\log(\frac{64N16^N}{\delta}),  \frac{2N}{\epsilon}(\frac{16N^2}{\epsilon})^{\frac{1}{\theta}}\right\}^{\frac{1}{1-\theta}}\right).
%\end{equation}
%It is easy to check that for $\|q^{mac}[t]\| \geq \exp\left((\frac{2N}{\epsilon})^{\frac{1}{1-\theta}}\log(1+e)\right)$, $w_{min}[t] \geq f(e)$ which satisfies (\ref{help3}). Therefore, obviously, (\ref{help3}) also holds for $q_{th}$ of (\ref{thresh2}).
%
%\end{remark}

The last step of the proof is to determine $t^*$. Let $t_1$ be the first time that $q_{max}(t)$ hits $q_{th}$, then
\begin{eqnarray*}
\sum_{k=t_1}^{t_1+t}\frac{1}{T_k^2} & \geq & 16^{-2|\mathcal{L}|} \sum_{k=t_1}^{t_1+t} e^{-16|\mathcal{L}| g(q_{max}(t))}\\
& = & 16^{-2|\mathcal{L}|} \sum_{k=t_1}^{t_1+t} e^{-16|\mathcal{L}| \frac{\log(1+q_{max}(t))}{h(q_{max}(t))}}\\
& = &16^{-2|\mathcal{L}|} \sum_{k=t_1}^{t_1+t} (1+q_{max}(t))^{-\frac{16|\mathcal{L}| }{h(q_{max}(t))}}\\
 & \geq & 16^{-2|\mathcal{L}|} t (1+q_{th}+t)^{-\frac{16|\mathcal{L}| }{h(q_{th})}}
\end{eqnarray*}
and
\begin{eqnarray*}
\min_{s}\pi_{t_1}(s)& \geq & \frac{1}{\sum_{s}\exp(\sum_{i \in s}\widetilde{w}_{ij}(t_1))}\\
& \geq & \frac{1}{|\mathcal{R}|\exp(|\mathcal{L}|\widetilde{w}_{max}(t_1))}\\
& \geq & \frac{1}{|\mathcal{R}|\exp(|\mathcal{L}|({w}_{max}(t_1)+g^*(t_1)))}\\
& \geq & \frac{1}{2^{N^2}\exp(2N^2 g(q_{th}))}
\end{eqnarray*}
Therefore, by Proposition \ref{drift}, it suffices to find the smallest $t$ that satisfies
\ben
16^{-2N^2} t (1+q_{th}+t)^{-\frac{16N^2 }{g(q_{th})}} &\geq& \log(4/ \delta) + N^2  \log(2(1+q_{th}))
\een
for a threshold $q_{th}$ large enough. Recall that $h(.)$ is an increasing function, therefore, by choosing $q_{th}$ large enough, $\frac{16N^2}{h(q_{th})}$ can be made arbitrary small.
Then a finite $t^*$ always exists since
$$
\lim _{t^* \to \infty} t^* (1+q_{th}+t^*)^{-\frac{16N^2}{h(q_{th})}} = \infty.
$$
\section{Proof of Proposition \ref{drift}}
The drift in $\pi_t$ is given by
\begin{eqnarray*}
\|\pi_{t+1}-\pi_t\|^2_{1/\pi_{t+1}} &=&\|\frac{\pi_{t}}{\pi_{t+1}}-1\|^2_{\pi_{t+1}} = \sum_s \pi_{t+1}(s)(\frac{\pi_{t}(s)}{\pi_{t+1}(s)}-1)^2\\
& \leq & \max \{(e^{\alpha_t}-1)^2,(1-e^{-\alpha_t})^2\} = (e^{\alpha_t}-1)^2
\end{eqnarray*}
for $\alpha_t <1$ where $\alpha_t$ is given by (\ref{eq: alpha}). Thus, $
\|\pi_{t+1}-\pi_t\|_{1/\pi_{t+1}} \leq 2 \alpha_t
$
for $\alpha_t <1$.
The distance between the true distribution and the stationary distribution at time $t$ can be bounded as follows.
First, by triangle inequality,
\begin{eqnarray*}
\|\mu_t-\pi_t\|_{1/\pi_{t}} &\leq & \|\mu_t-\pi_{t-1}\|_{1/\pi_{t}}+\|\pi_{t-1}-\pi_{t}\|_{1/\pi_{t}}\leq  \|\mu_t-\pi_{t-1}\|_{1/\pi_{t}} + 2 \alpha_{t-1}.
\end{eqnarray*}
On the other hand,
\begin{eqnarray*}
\|\mu_t-\pi_{t-1}\|^2_{1/\pi_{t}} & =& \sum_{s}\frac{1}{\pi_t(s)}(\mu_t(s)-\pi_{t-1}(s))^2 =  \sum_{s}\frac{\pi_{t-1}(s)}{\pi_t(s)}\frac{1}{\pi_{t-1}(s)}(\mu_t(s)-\pi_{t-1}(s))^2\\
& \leq & e^{\alpha_{t-1}} \|\mu_t-\pi_{t-1}\|^2_{1/\pi_{t-1}}.
\end{eqnarray*}
Therefore, for $\alpha_{t-1} < 1$,
\begin{eqnarray*}
\|\frac{\mu_t}{\pi_t}-1\|_{\pi_t} & \leq & (1+\alpha_{t-1})\|\mu_t-\pi_{t-1}\|_{1/\pi_{t-1}}+2 \alpha_{t-1}.
\end{eqnarray*}
Suppose $\alpha_t \leq \delta/16$, then $\|\frac{\mu_t}{\pi_t}-1\|_{\pi_t} \leq \delta/2$ holds for $t> t^*$, if
$$
\|\mu_t-\pi_{t-1}\|_{1/\pi_{t-1}} \leq \delta/4
$$
for all $t > t^*$. Define $a_t:= \|\mu_{t+1}-\pi_{t}\|_{1/\pi_{t}}$. Then
\begin{eqnarray*}
a_{t+1} &=& \|\mu_{t+2}-\pi_{t+1}\|_{1/\pi_{t+1}} =\|\mu_{t+1}P_{t+1}-\pi_{t+1}\|_{1/\pi_{t+1}}\leq  \lambda^*_{t+1}\|\mu_{t+1}-\pi_{t+1}\|_{1/\pi_{t+1}}
\end{eqnarray*}
where $\lambda^*_{t+1}$ is the SLEM of $P_{t+1}$.
Therefore,
$$
a_{t+1} \leq \lambda^*_{t+1} [(1+\alpha_t) a_t + 2 \alpha_t].
$$
Suppose $a_{t} \leq \delta/4$. Defining $T_t=\frac{1}{1-\lambda^*_t}$, we have
$$
a_{t+1} \leq (1-\frac{1}{T_{t+1}})[\delta/4+(2+\delta/4)\alpha_t].
$$
Thus, $a_{t+1} \leq \delta/4$, if
$$
(2+\delta/4)\alpha_t < \frac{1}{T_{t+1}}(\delta/4+(2+\delta/4)\alpha_t),
$$
or equivalently if $
\alpha_t < \frac{\frac{\delta/4}{T_{t+1}}}{(2+\delta/4)(1-1/T_{t+1})}.
$
But
$$
\frac{\frac{\delta/4}{T_{t+1}}}{(2+\delta/4)(1-1/T_{t+1})} > \frac{\frac{\delta/4}{T_{t+1}}}{4(1-1/T_{t+1})} > \frac{\delta}{16}\frac{1}{T_{t+1}},
$$
so, it is sufficient to have
\be \label{eq: alphaT2}
\alpha_t T_{t+1} \leq \delta/16.
\ee
Therefore, if there exists a time $t^*$ such that $a_{t^*} \leq \delta/4$, then $a_{t} \leq \delta/4$ for all $t \geq t^*$.
To find $t^*$, note that $a_t > \delta/4$ for all $t< t^*$. So, for $t<t^*$, we have
\begin{eqnarray*}
a_t &\leq & (1-\frac{1}{T_t})[(1+\alpha_{t-1})a_{t-1}+2\alpha_{t-1}]\leq  (1-\frac{1}{T_t})[(1+\alpha_{t-1})a_{t-1}+2\alpha_{t-1}4\frac{a_{t-1}}{\delta}]\\
& \leq & (1-\frac{1}{T_t})(1+\alpha_{t-1}+\frac{8}{\delta}\alpha_{t-1})a_{t-1}\leq (1-\frac{1}{T_t})(1+\frac{\delta/16}{T_t}(1+\frac{8}{\delta}))a_{t-1}\\
& \leq & (1-\frac{1}{T_t})(1+\frac{1}{T_t}) a_{t-1}= (1-\frac{1}{T^2_t})a_{t-1} \leq  e^{-\frac{1}{T^2_t}}a_{t-1}.
\end{eqnarray*}
Thus, $
a_t \leq a_0 e^{-\sum_{k=1}^{t^*}\frac{1}{T^2_k}},
$
 where
 \begin{eqnarray*}
a_0 &=&  \|\frac{\mu_1}{\pi_0}-1\|_{\pi_0}  =  \|\mu_0P_0-\pi_0\|_{1/ \pi_0}\\
&\leq&  \lambda^*(P_0)\|\mu_0-\pi_0\|_{1/ \pi_0}  \leq  \sqrt{\frac{1}{\min_s\pi_0(s)}}.
 \end{eqnarray*}
Finally, assume that (\ref{eq: alphaT2}) holds only when $q_{max}(t) \geq q_{th}$ for a constant $q_{th} > 0$. Let $t_1$ be the first time that $q_{max}(t)$ hits $q_{th}$.
Then, after that, it takes $t^*$ time slots for the chain to get close to $\pi_t$ if $q_{max}(t)$ remains above $q_{th}$ for $t_1 \leq t \leq t_1+t^*$.
Alternatively, we can say that $\|\pi_t-\mu_t\|_{TV} \leq \delta/4$ if $q_{max}(t) \geq q_{th}+t^*$ since at each time slot at most one departure can happen and this guarantees that $q_{max}(t) \geq q_{th}$ for, at least, the past $t^*$ time slots. This immediately implies the final result in the proposition.
\end{document}